\documentclass[a4paper,12pt]{article}
\usepackage[margin=1in]{geometry}  

\usepackage{listings}
\usepackage{amsmath,mathtools}
\usepackage{amsthm}
\usepackage{tikz}
\usepackage{caption,subcaption}
\usepackage{array}
\usepackage{mdwmath}
\usepackage{multirow}
\usepackage{mdwtab}
\usepackage{eqparbox}
\usepackage{multicol}
\usepackage{amsfonts}
\usepackage{multirow,bigstrut,threeparttable}
\usepackage{amsthm}
\usepackage{array}
\usepackage{bbm}
\usepackage{subcaption}
\usepackage{epstopdf}
\usepackage{mdwmath}
\usepackage{mdwtab}
\usepackage{eqparbox}
\usepackage{tikz}
\usepackage{latexsym}
\usepackage{amssymb}
\usepackage{bm}
\usepackage{amssymb}
\usepackage{graphicx}
\usepackage{mathrsfs}
\usepackage{epsfig}
\usepackage{psfrag}
\usepackage{setspace}
\usepackage{tikz-cd}

\usepackage[
CJKbookmarks=true,
bookmarksnumbered=true,
bookmarksopen=true,
colorlinks=true,
citecolor=red,
linkcolor=blue,
anchorcolor=red,
urlcolor=blue
]{hyperref}
\usepackage{cleveref}
\usepackage{algorithm}
\usepackage{algorithmic}
\usepackage{stfloats}
\usepackage[
autocite    = superscript,
backend     = biber, 
sortcites   = true,
style       = apa
]{biblatex} 
\usepackage{nameref}
\usepackage[normalem]{ulem}
\newcommand{\zg}[1]{{\color{blue} [ZG: #1]}}

\usepackage{makecell}

\newcommand{\wh}[1]{{\color{cyan} [WH: #1]}}

\usepackage{soul}
\newcommand{\qz}[1]{{\color{red} [QZ: #1]}}

\input xy
\xyoption{all}

\theoremstyle{plain}
\newtheorem{theorem}{Theorem}
\newtheorem{proposition}{Proposition}
\newtheorem{lemma}{Lemma}
\newtheorem{corollary}{Corollary}

\crefname{theorem}{Theorem}{Theorems}
\crefname{proposition}{Proposition}{Propositions}
\crefname{corollary}{Corollary}{Corollaries}

\theoremstyle{definition}
\newtheorem{definition}{Definition}

\newtheorem{example}{Example}

\crefname{definition}{Definition}{Definitions}
\crefname{remark}{Remark}{Remarks}

\def\EE{\mathbb{E}}

\def\NN{\mathbb{N}}

\def\PP{\mathbb{P}}

\def\RR{\mathbb{R}}

\def\calA{\mathcal{A}}

\def\calD{\mathcal{D}}

\def\calI{\mathcal{I}}

\def\calN{\mathcal{N}}

\def\calS{\mathcal{S}}

\def\calX{\mathcal{X}}

\newcommand\pr{{\mathbb{P}}} 

\def\1{\mathbbm{1}}

\newcommand\independent{\protect\mathpalette{\protect\independenT}{\perp}}
\def\independenT#1#2{\mathrel{\rlap{$#1#2$}\mkern2mu{#1#2}}}

\newcommand{\argmin}{\mathop{\mathrm{argmin}}}

\usepackage{booktabs}
\usepackage{adjustbox}
\usepackage{multirow}
\usepackage{listings}

\usepackage{authblk}

\theoremstyle{plain}

\usepackage{xspace}

\def\independenT#1#2{\mathrel{\rlap{$#1#2$}\mkern2mu{#1#2}}}

\definecolor{myblue}{rgb}{.8, .8, 1}
\definecolor{mathblue}{rgb}{0.2472, 0.24, 0.6} 
\definecolor{mathred}{rgb}{0.6, 0.24, 0.442893}
\definecolor{mathyellow}{rgb}{0.6, 0.547014, 0.24}

\usepackage{enumitem}

\newcommand{\No}{{m}}

\newcommand{\FDR}{\text{FDR}}

\newcommand{\NoTask}{{k}}

\usepackage{apptools}
\AtAppendix{\counterwithin{theorem}{section}}
\AtAppendix{\counterwithin{proposition}{section}}
\AtAppendix{\counterwithin{definition}{section}}
\AtAppendix{\counterwithin{example}{section}}
\AtAppendix{\counterwithin{lemma}{section}}
\AtAppendix{\counterwithin{equation}{section}}

\def\showComment{0}

\title{\LARGE \bf A constructive approach to selective risk control}

\def\blind{0} 

\if\blind0
\author[1]{Zijun Gao}
\author[2]{Wenjie Hu}
\author[3]{Qingyuan Zhao}
\affil[1]{Department of Data Science and Operations, University of Southern California, USA}
\affil[2]{Department of Biostatistics, Epidemiology \& Informatics, University of Pennsylvania, USA}
\affil[3]{Statistical Laboratory, University of Cambridge, UK}
\fi

\if\blind1
\author[1]{}
\fi

\date{}

\begin{document}

\maketitle
\if\blind1
\vspace{-2cm}
\fi
\begin{abstract}
  Many modern applications require using data to select
  the statistical tasks and make valid inference after
  selection. In this article, we provide a unifying approach to
  control for a class of selective risks. Our method is
  motivated by a reformulation of the celebrated Benjamini--Hochberg
  (BH) procedure for multiple hypothesis testing as the fixed point iteration of the  
  Benjamini--Yekutieli (BY) procedure for constructing post-selection
  confidence intervals. 
  Building on this observation, we propose a
  constructive approach to control extra-selection risk (where selection is
  made after decision) by iterating decision strategies that control
  the post-selection risk (where decision is made after selection), and show
  that many previous methods and results are special cases of this
  general framework. 
  Our development leads to two surprising results about the
  BH procedure: (1) in the context of one-sided
  location testing, the BH procedure not only controls the false
  discovery rate at the null but also at other locations \emph{for
    free}; (2) in the context of permutation tests, the BH procedure
  with exact permutation p-values can be well approximated by a
  procedure which only requires a total number of permutations that is
  almost \emph{linear} in the total number of hypotheses.

\end{abstract}

\noindent%
{\it Keywords:}  selection inference, multiple testing, false discovery rate, fixed-point iteration, permutation test 



\section{Introduction} \label{sec:intro}

In classical statistical theory, it is assumed that the
scientific hypotheses of interest and the statistical model are
determined before data analysis. In modern applications, however, data analysts
are often interested in using the data to select suitable models or
hypotheses to be tested. It is widely recognized that classical
inference procedures are usually invalid in presence of
selection. This has led to a vast literature on ``selective
inference''. Under this umbrella term are two broad types of problems:
\begin{description}
    \item[The selection problem:] How can we select promising signals for further
      investigation? 
    \item[The post-selection problem:] How can we make statistical
      inference after some tasks have been selected using the data? 
\end{description}

Although these two types of problems are often studied in
isolation in the literature, they are clearly related. For one thing,
it is natural to expect that good post-selection inference should take
into account of how the selection is made. Moreover, the selection
problem can be viewed as a special case of the post-selection problem,
where the target of inference is the selection itself. 

This article is motivated by a surprising observation that connects the celebrated Benjamini–Hochberg (BH) procedure for controlling the false discovery rate (FDR)---a multiple testing method for the selection problem---and the Benjamini-Yekutieli procedure for controlling the false coverage-statement rate (FCR)---a well-known method for the post-selection problem. In the original article by \textcite{benjamini2005false} that introduces the BY procedure, it is already observed that the output of the BH procedure is a \emph{fixed point} of the BY procedure, that is,
\begin{equation*} 
  \text{BH}(\text{data}) = \text{BY} \circ \text{BH} (\text{data}),
\end{equation*}
where $\circ$ means composition of those procedures (a precise definition will be given later). With this observation, \textcite[section 6.1]{benjamini2005false} uses the FCR controlling property of BY to prove the FDR controlling property of BH. We find that the BH and BY procedures have an even closer connection: BH is exactly the \emph{fixed point iteration} of BY in the sense that
\begin{equation} \label{eq:bh-by-iteration}
  \text{BH}(\text{data}) = \underbrace{\text{BY} \circ \text{BY} \circ \dots \circ \text{BY}}_{\text{until convergence}} (\text{data}).
\end{equation}
We next use a numerical example to clarify what we mean by \eqref{eq:bh-by-iteration} and how it motivates the rest of this article.

\subsection{The BH procedure as a fixed point iteration}
\label{sec:reform-benj-hochb} \label{sec:illustration}

\if\showComment1
\wh{1. iterative procedure. 2. The connection between FCR and FDR: The risk we want to control is FCR. When the selected parameters' confidence intervals don't contain 0 (this is a constraint we need for the selected tasks). When the FCR is controlled under this constraint selected tasks, the FDR is also controlled. 3. Computation acceleration compared to standard BH. Use figures to illustrate.}
\fi

Suppose a data analyst observed $\No$ independent z-statistics, $X_i \sim
\calN(\theta_i, 1)$ and the associated one-sided p-values $P_i = \Phi(X_i)$ for the null hypotheses $H_{i}: \theta_i \ge 0$, $i = 1,\dots,\No$, where $\Phi$ is the cumulative distribution
function of the
standard normal distribution. Let $X_{(1)} \leq \dots \leq X_{(\No)}$
be the order statistics and $P_{(1)} \leq \dots \leq P_{(\No)}$ be the
associated p-values. Given $q\in [0,1]$, the BH procedure rejects the null hypothesis
$H_i:\theta_i \geq 0$ if
\begin{equation}
  \label{eq:bh}
  P_i \leq P^{*},~\text{where}~P^{*} = P_{(I^{*})}~\text{and}~I^{*} =
  \max \left \{i: P_{(i)} \leq (i/\No) q \right \},
\end{equation}
and controls the FDR at level $q$, meaning that
\begin{equation}
  \label{eq:fdr}
  \FDR = \EE\left[ \frac{\sum_{i=1}^{\No} S_i^{*} 1_{\{\theta_i \ge
      0\}}}{1 \vee \sum_{i=1}^{\No} S_i^{*}}  \right] \leq q,
\end{equation}
where $S_i^{*} = 1_{\{P_i \leq P^{*}\}}$. The numerator and
denominator in \eqref{eq:fdr} correspond, respectively, to the number
of false rejections and the total number of rejections (except that here we take the maximum with 1 to avoid issues with $0/0$).

The Benjamini--Yekutieli (BY) procedure is a general method that
addresses the post-selection problem. Suppose a subset $\mathcal{S}
\subseteq \{1,\dots,m\}$ of the indexes of parameters is selected using the data,
and one is interested in constructing confidence intervals for the
selected parameters. Let $S_i = 1_{\{i \in \mathcal{S}\}}$ be the
indicator that $i$ is selected. When specializing to the normal
location problem considered here, the BY procedure computes the
 one-sided confidence intervals $(-\infty, U_i)$ for all $i
\in \mathcal{S}$, where
\begin{equation}
  \label{eq:BY}
  U_i =  X_i + \Phi^{-1}(1 - q |\mathcal{S}| / \No).
\end{equation}
\textcite{benjamini2005false} showed that if the selection procedure
that gives $\mathcal{S}$ is ``stable'' in a sense defined in \Cref{sec:pre-post} below, then the confidence intervals
controls the FCR, that is,
\begin{equation}
  \label{eq:fcr}
  \mathrm{FCR} = \EE\left[ \frac{\sum_{i=1}^{\No} S_i 1_{\{\theta_i \ge
      U_i\}}}{1 \vee \sum_{i=1}^{\No} S_i}  \right] \leq q.
\end{equation}

A careful reader may have already noticed the similarities between the
BH and BY procedures in how the nominal level $q$ is adjusted and the
expression of the error rates. But as we will show next, there is an
even closer connection: the BH procedure is \emph{exactly} the
fixed-point iteration limit of the BY procedure!

\begin{figure}[t]
        \centering
        \begin{minipage}{1\textwidth}
                \centering
                \includegraphics[clip, trim = 0cm 0cm 0cm 0cm, height = 10cm]{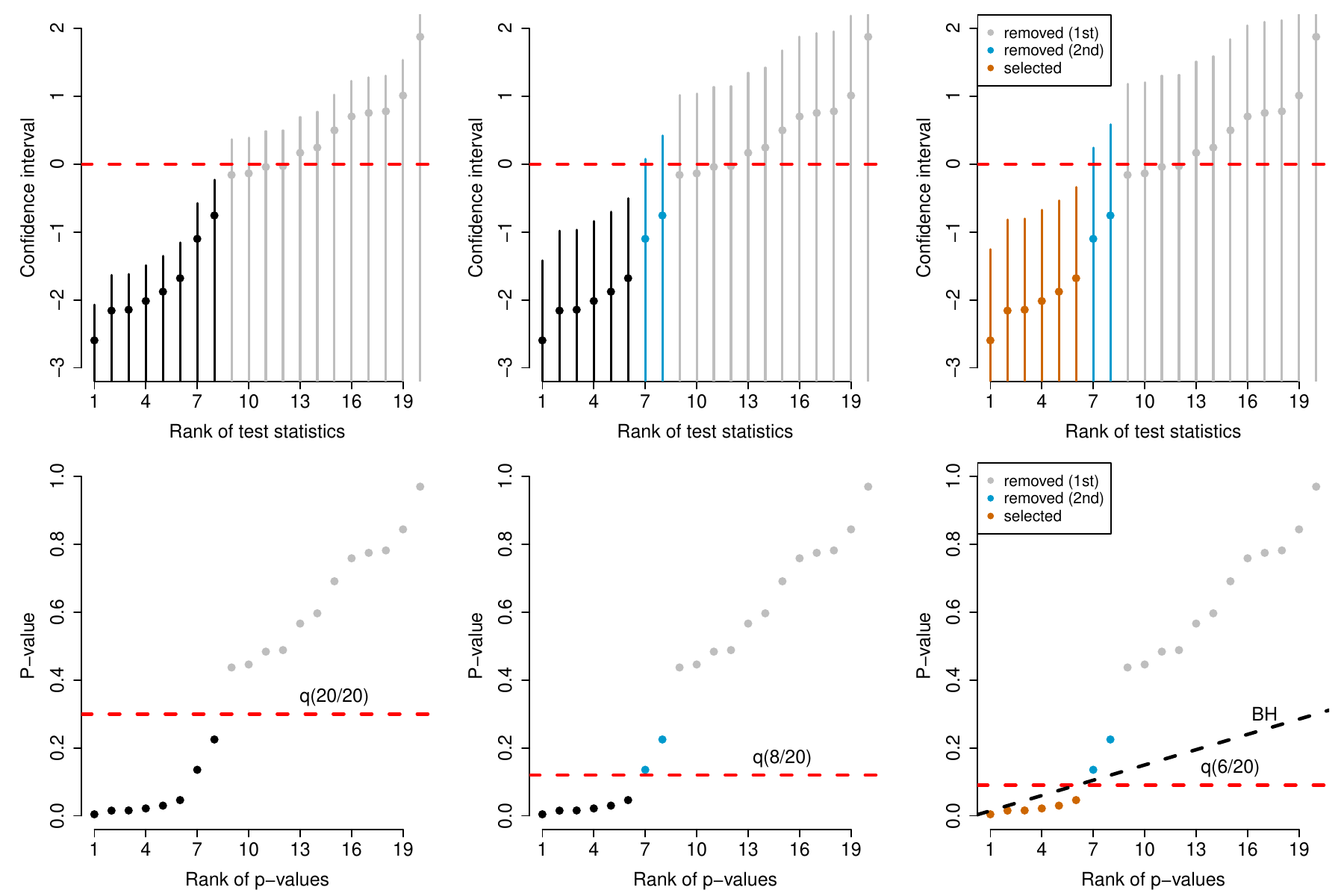}
                \subcaption*{\qquad 1st iteration \qquad\qquad\qquad\qquad\quad 2nd iteration \qquad\qquad\qquad\qquad 3rd iteration }
        \end{minipage}
       \caption{The BH procedure as the fixed point iteration of the BY procedure. In this example, we are interested in testing $H_i: \theta_i \ge 0,~i=1,\dots,20$ using $X_i \sim \mathcal{N}(\theta_i,1)$.
         In the first iteration (left panels), $20$ one-sided confidence intervals are constructed at level $1 - q = 0.7$, and $12$ hypotheses are removed because the corresponding confidence intervals cover zero (top left panel).
         This is equivalent to eliminating the hypotheses with p-values $\geq q = 0.3$ (top right panel).
         In the second iteration (middle panels), $8$ one-sided confidence intervals at level $1-q (8/20)$ are constructed and $2$ more hypotheses 
         are removed.
         This is equivalent to de-selecting the hypotheses with p-values $\geq (8/20)q$.
         In the third iteration (right panels), $6$ one-sided confidence intervals at level $1-(6/20) q$ are constructed and none of them covers zero; in other words, all remaining p-values are below $(6/20)q$.
         In the bottom right panel, a reference line with slope $q/20$ is superimposed, showing the BH procedure rejects exactly those 6 hypotheses.
         }
        \label{fig:iterate.FCR}
\end{figure}

We first explain what we mean by ``iterative limit'' with the help of
a numerical example in \Cref{fig:iterate.FCR} (more data about this example are provided in \Cref{tab:iterative.FCR} in the appendix). Imagine the data
analyst, having observed $X_1 = -2.59,\dots,X_{20} = 1.88$  is interested in
identifying the negative values among $\theta_1,\dots,\theta_{20}$. She
first constructed 20 (unadjusted) one-sided confidence intervals, and
decided to only keep those that do
not cover $0$; this is illustrated by the top left panel of
\Cref{fig:iterate.FCR}.
Initially, the analyst was content with just reporting
these $8$ confidence intervals, but was then warned by another statistician
about the perils of selection bias. Following the advice of that
statistician, the analyst applied the BY procedure
and obtained adjusted confidence intervals for the selected
parameters, illustrated by the middle panel on the top row of
\Cref{fig:iterate.FCR}. She was told by the statistician that the new confidence intervals are valid in the sense of \eqref{eq:fcr}.

But there is an obvious problem: unlike the first unadjusted
confidence intervals, some of the new adjusted confidence intervals
now cover $0$, and the data analyst was tempted to just report the $6$
new confidence intervals 
that do not
cover $0$. Of course, there is another catch: the new confidence
intervals may no longer be valid due to the
\emph{extra selection} (from $8$ to $6$) being made. This time, the
data analyst recognized a clever solution: she can just apply the BY
procedure to the $6$ selected parameters once again! She found that
although the confidence intervals become wider with the additional
adjustment, none of them covers $0$ (top right panel of
\Cref{fig:iterate.FCR}), and she delightedly shared these results with
her colleagues.

The next proposition states that these $6$ parameters are exactly those
that would be rejected by the BH procedure. This is also illustrated in bottom right panel of \Cref{fig:iterate.FCR}.

\begin{proposition} \label{prop:fcr-to-fdr}
  In the one-sided normal location problem described above, let
  $\mathcal{S}^1 = \{1,\dots,m\}$ and consider the fixed point iteration of the BY procedure: 
  \begin{align*}
    U_i^t = X_i + \Phi^{-1}(1 - q |\mathcal{S}^{t}| / \No),~i \in
            \mathcal{S}^t \quad \text{and} \quad
    \mathcal{S}^{t+1} = \{i: U_i^t \leq 0\}, \quad t = 1, 2,\dots.
  \end{align*}
  This iterative procedure converges in finite time, meaning $T =\inf\{t
  \geq 1:S_{t-1} =S_t\} - 1 < \infty$ and $\mathcal{S}^t =
  \mathcal{S}^{T}$ for all $t \geq T$. Moreover, the selected set at
  convergence is given by
  \[
    S^T = \{i: P_i \leq P^{*}\},
  \]
  where $P^{*}$ is the p-value threshold determined by the BH procedure
  given in \eqref{eq:bh}.
\end{proposition}

\if\showComment1
\qz{Selective inference is important. What does selective inference broadly mean? What are some common kinds of selective inference
problems? Two types of decisions: (1) selection (FWER, FDR, etc.); (2)
More general inference after selection (conditional post-selection;
PoSI, debiased-LASSO, FCR). (3) Third type (not much work): both the selection and
post-selection inference are of interest.}
\fi

\if\showComment1
\wh{(1) and (2) can be seen as special cases of (3). In the third type, sometimes the selection may depend on the inference results and the inference results also depend on the selection.
 Then it is natural to consider  iterating the selection procedure  and the inference procedure to simultaneously perform selection and post-selection inference. When the post-selection inference is in fact the same as selection, (3) degenerates to (1); When the selection doesn't depend on the inference results, that is, we don't need to change our selection after inference, (3) degenerates to (2).  (The aim of third type is to control inference risk (multiple inference risks in \Cref{sec:FDR.curve}) and make the selection satisfy some constraints/properties )}
\zg{Perhaps move \Cref{fig:iterative} here.}
\fi

\if\showComment1
\qz{What we focus on: selective risk control. In
particular, select a subset of statistical task using the data, and the selective
risk we use is the expected average loss for the decisions on the
selected tasks.}
\fi

\if\showComment1
Next \zg{Perhaps add this in \Cref{sec:illustration}}: Start with a ``surprising'' re-formulation of Benjamini-Hochberg, a
multiple testing procedure that controls FDR, via \cite{benjamini2005false}, a post-selection inference procedure that
controls FCR.
\fi

\subsection{Related work, contributions, and organization}
\label{sec:overview-article}

Motivated by the observation about the BH procedure in \eqref{eq:bh-by-iteration}, the rest of this article develops a constructive approach to selective risk control in compound decision problems. 
Before outlining that approach, let us first briefly review some related work because similar ideas have been discussed in the multiple testing literature for a long time. 

As mentioned earlier, \textcite[section 6.1]{benjamini2005false} used the fact that the output of BH is a fixed point of BY to prove the FDR controlling property of BH. \textcite{weinstein2020selective} consider combining the BY procedure with a selection criterion to construct selective confidence intervals that meet specific criteria, such as not-covering-zero.
In another
closely related work, \textcite{blanchard2008two} provided a set of
two sufficient conditions for FDR control; among them is a fixed point
condition termed ``self-consistency''. \textcite{madar16_fastl} showed
that the BH procedure has a fast implementation through a basic iteration, but they did not seem to realize that
basic iteration is precisely the BY procedure. \textcite{barberPfilterMultilayerFalse2017}
proposed a BH-like threshold that can control multiple FDRs
simultaneously and described an iterative algorithm that is very much like the coordinate-wise (each FDR is considered as a coordinate) generalization of the fixed point iteration
described above. 

However, we could not find any previous work that
explicitly points out that \emph{BH is precisely the fixed-point
iteration of BY}. In informal discussion, it was remarked to us that BH and BY have a ``dual'' relationship, but that description is not accurate. Hypothesis testing (the task of BH) and constructing confidence intervals (the task of BY) have a duality in classical statistical theory \parencite{neyman1937outline}, but the precise connection between BH and BY is the fixed-point iteration in \eqref{eq:bh-by-iteration}. 

The main purpose of this article is to extend this fixed-point iteration perspective of BH to a much wider class of problems.
More specifically, we make the following contributions in the rest of this article:
\begin{enumerate}
    \item We extend the above reformulation of the BH procedure to general compound decision problems. This general perspective allows us to unify many existing procedures in multiple testing and post-selection inference in a single framework.
    
    \item We use this framework to develop new procedures that control multiple selective risks at the same time. Through the lens of this new framework, we find a surprising "free lunch theorem" that shows the original BH procedure not only controls the FDR for the standard null hypotheses $H_i: \theta_i \geq 0$ but also the FDR for $H_i: \theta_i \geq c$ for a range of $c$ at other locations (albeit at different levels).
    
    \item The iterative perspective can be used to accelerate the BH procedure. The acceleration is particularly useful when each p-value is computationally expensive. 
    We show that if each p-value is obtained by a permutation test, a total number of $O(\No \log^2\No)$ ($\No$ is the number of hypothesis) permutations would suffice to provide accurate approximations to the BH procedure with exact permutation p-values. 
\end{enumerate}

In \Cref{sec:constr-appr-post} below, we describe the formal compound decision
framework and three types of risks: the non-selective or pre-selection
risk (e.g.\ type I error), the post-selection risk (e.g.\ FCR), and
the extra-selection risk (e.g.\ FDR). These new concepts will be introduced and distinguished in the context of a game between a selection agent and a decision agent that generalizes the observation in \Cref{sec:reform-benj-hochb}.
We then show how a procedure that controls the extra-selection risk
can be constructed by iterating a procedure that controls the post-selection risk. 
In \Cref{sec:multiple-risks}, we use this constructive approach to develop
new methods that control multiple (or even an infinite number of) selective risks at
the same time. 
This requires an extension of the game described in
\Cref{sec:constr-appr-post} to multiple selection and decision agents. 
In \Cref{sec:accelerate.permutation}, we discuss how the  observation in \Cref{sec:reform-benj-hochb} can
be used to accelerate the BH procedure when the p-values are obtained from expensive permutation tests.
In \Cref{sec:discussion}, we will conclude the article with some discussion on future research directions. Technical proofs can be found in Supplementary Materials. Replication code for the real data analyses in \Cref{sec:protein-data,sec:cancer.data} can be found at \url{https://github.com/ZijunGao/A-constructive-approach-to-selective-risk-control-in-compound-decision-problems}. We provide an R Shiny App that computes the improved FDR curve in \Cref{sec:FDR.curve} (\url{https://zijungao.shinyapps.io/FDRCurve/}).

\section{A constructive approach to selective risk control}
\label{sec:constr-appr-post}
In this section,  we first introduce three selective risks: \emph{pre-selection
  risk}, \emph{post-selection risk} and \emph{extra-selection risk} in a decision theoretical framework in \Cref{sec:non-pre-post}. Then we give a constructive approach to selective risk control by the following two steps:
\begin{enumerate}
\item use individual decision rules that control the pre-selection
  risk to design a method that controls the post-selection risk (``pre
  to post''), details in  \Cref{sec:pre-post} ;
\item use fixed point iteration of the post-selection method from the previous step to control the extra-selection risk  (``post to extra''), details in \Cref{sec:post-extra}.
\end{enumerate}

\subsection{Selective risks in compound decision problems}
\label{sec:non-pre-post}

We first describe a compound decision setting that does not involve
adaptive selection using the data. Suppose we have observed some data
$\bm X$ taking values in $\calX$ from a distribution $P_\theta$, where $\theta \in \Theta$ is an unknown parameter.
Suppose further that  we have $m$ statistical tasks.
For task $i$, a decision rule $d_i: \calX \rightarrow
\calD_i$ is defined as a procedure that takes in the data $\bm X$ and
returns a decision $D_i = d_i(\bm X)$ in a decision space $\calD_i$. Each decision is evaluated
based on a loss function $\ell_i: \calD_i \times \Theta \rightarrow
[0,1]$. 
In this paper we assume the loss is
bounded between $0$ and $1$.
The \emph{risk} of a decision procedure $d_i$ is
defined as
\begin{align}\label{eq:non.selective.risk.individual}
  r_i(d_i,\theta) := \EE_{\theta}\left[\ell_i(d_{i}(\bm X),  \theta)\right],
\end{align}
where $\EE_{\theta}$ denotes  the expectation with respect to $\bm X
\sim P_{\theta}$.

Often, we are interested in controlling the risk at
some level. We say a collection of decision rules
$\{d_i(\cdot;q): q \in [0,1]\}$ controls the $\ell_i$-risk
if its risk is  bounded above by the target level
$q$ under any distribution $P_\theta$, that is,
\begin{align}\label{eq:control}
    \sup_{\theta \in \Theta} r_i(d_i(\cdot; q), \theta)
    \le q, \quad \forall~ q \in [0,1].
\end{align}

\begin{example}[Multiple testing] \label{exam:fwer}
  In multiple testing, the $i$-th statistical task is to test a null hypothesis $H_i:\theta \in \Theta_i$, and the decision space $\calD_i$ is $\{0,1\}$, where $D_i = 1$ means the null hypothesis $H_i$ is rejected. The loss function is given by
  \begin{align*}
    \ell_i(d_i,\theta) =
    \begin{cases}
    1, & \text{if}~\theta \in \Theta_i~\text{and}~d_i = 1,\\
    0, & \text{otherwise}.
    \end{cases}
  \end{align*}
  In the simplest setting, the $i$-th test is done by using a single
  p-value $P_i$. The definition that the decision
  procedure $d_i$ controls the (non-selective) risk essentially means
  the p-value $P_i$ is valid in the sense that $\PP_{\theta}(P_i \leq
  q) \leq q$ for all $\theta \in \Theta_i$ and $q \in
  [0,1]$. In simultaneous hypothesis testing, a commonly-used measure
  of risk is the family-wise error
  rate $\text{FWER}=\EE_\theta [\max_{i} \ell_i(d_i,\theta)]$.
\end{example}

\begin{table}[tbp]
    \centering
    \begin{tabular}{p{5cm}|l}
    \toprule
    Notation & Description\\
    \midrule
    $\bm X \in \calX$         & data \\
    $\theta \in \Theta$     & parameter of the distribution of $\bm X$  \\ \midrule
    $d_i: \calX \to \calD_i$     & decision rule of task $i$  \\
    $D_i = d_i(\bm X)$     & decision of task $i$  \\
    $\tilde{d}_i: \{0,1\}^m \times \mathcal{X} \to \mathcal{D}_i$     & decision strategy of task $i$ \\ \midrule
    $s_i: \calX \to \{0,1\}$ or & selection rule or extended selection rule    \\
    $s_i: \{0,1\}^{\No} \times \calX \to \{0,1\}$ & \\
    $S_i = s_i(\bm X)$ or $s_i(\bm w, \bm X)$ & selection indicator: $S_i = 1$ means task $i$ is selected \\ 
    $\tilde{s}_i: \mathcal{D} \times \mathcal{X} \to \{0,1\}$ & selection strategy \\
    $\bm S = (S_i)$  & selection indicator vector \\
    $\calS \in 2^{[\No]}$     & set of selected tasks \\
    $|\calS|$ or $|\bm S|$     & number of selected tasks\\     $\bm S \preceq \bm S'$ & $S_i \le S_i'$, $i \in [\No]$ \\ \midrule
    $\ell_i(D_i, \theta)$     & loss of task $i$ associated with the decision $D_i$ \\
    $r(\bm d, \bm s, \theta)$     & risk associated with decision rule $\bm d$ and selection rule $\bm s$ \\
    $r^{\mathrm{pre/post/extra}}(\bm d, \theta)$    & pre/post/extra-selection risk of decision rule $\bm d$ \\        \bottomrule
    \end{tabular}
    \caption{Notation in this article. Decision rule/strategy share the same domain and codomains as the selection rule/strategy and are denoted using $d$ instead of $s$.}
    \label{tab:notations}
\end{table}

In this paper, we are interested in the problem of selective
inference, where one may use the data to select which decisions
could or should be made. In other words, the risk is a function of
both the decision rule and a selection rule $\bm s: \calX \to
\{0,1\}^m$. A selective inference problem is then characterized by how
the risk takes the selection rule $\bm s$ into account. In this paper,
we consider the following risk function that averages the losses among
the selected tasks:
\begin{equation}
  \label{eq:selective.risk}
  r(\bm d, \bm s, \theta) := \EE_{\theta}\left[ \frac{\sum_{i \in \mathcal{S}}
    \ell_i(D_i, \theta)}{1 \vee |\mathcal{S}|} \right] = \EE_{\theta}\left[
  \frac{\sum_{i=1}^m s_i(\bm X) \ell_i(d_i(\bm X),\theta)}{1 \vee \sum_{i=1}^m
    s_i(\bm X)} \right],
\end{equation}
where $\bm D = \bm d(\bm X) \in \calD := \calD_1 \times \dots \times
\calD_m $, $\bm S = \bm s(\bm X) = (s_1(\bm X), \ldots, s_m(\bm X))$, and $\mathcal{S} = \{i
\in [m]: S_i := s_i(\bm X) = 1\}$ is the set of selected tasks. This does not
include the FWER
in \Cref{exam:fwer} because FWER takes the maximum instead of the
average of the incurred losses.

When no data-adpative selection is made, that is $\mathcal{S} = [m]$, we say $r(\bm d,
\bm s, \theta)$ in \eqref{eq:selective.risk} is
a \emph{pre-selection} risk and denote it as
$r^{\mathrm{pre}}$. Controlling the
pre-selection risk is straightforward if the
individual decision rules control the corresponding risks. That is,
if \eqref{eq:control} is true for all $i \in [m]$, making the
same decisions will control the pre-selection risk because
\[
  r^{\mathrm{pre}}(\bm d, \theta) = r(\bm d, \bm s, \theta) =
  \frac{1}{m}
\sum_{i=1}^m r(d_i(\cdot; q), \theta) \leq
\frac{1}{m} \sum_{i=1}^m q = q.
\]

Without any detail about the selection rule $\bm s$ and the
decision rule $\bm d$, not much can be said about the selection risk
other than it can be much larger than the nominal level if no
adjustment is made. As an example, consider the multiple
testing problem in \Cref{exam:fwer} with decision rule
$d_i = \1_{\{P_i \leq q\}}$, which controls the non-selective
risk at $q$ if $P_i$ is a valid p-value for the $i$-th null
hypothesis. However, if we average the loss among
the rejected hypotheses, that is, if we let $\mathcal{S} = \{i \in
[m]: D_i = 1\}$ (so $\bm s = \bm d$), then
\[
  r(\bm d, \bm s, \theta) = \EE_{\theta} \left[
    \frac{|\mathcal{S}|}{|\mathcal{S}|
      \vee 1} \right] = \PP_{\theta}(|\mathcal{S}| \geq
  1)
\]
when all the null hypotheses are true (so the numerator is
$|\mathcal{S}|$). This converges to $1$ if the
individual p-values are independent and the number of hypotheses
converges to infinity.

We further distinguish selective inference problems by the order of
selection and decision (note that they are not necessarily
ordered in general). One type of selective inference problems assume that the
decision is made after the selection; we will call this the
``post-selection'' problem. Another type of problems assume
that the selection is made after the decision (or the decision itself
is about selection); we will call this the ``extra-selection''
problem. To understand the different between post- and extra-selection problems,
consider a game betwen two agents called Ganon and Link
(\Cref{fig:one-iteration}). The first agent, Ganon, selects a set
of tasks $\mathcal{S}^1$, and the second agent, Link, uses the data to
perform those tasks and make decisions $\bm D^1$. After receiving
Link's decision, Ganon may further select a subset $\mathcal{S}^2
\subseteq \mathcal{S}^1$. In each step, both agents are allowed to use
the data $\bm X$. Two notions of selection risks can be defined according
to which selected set of tasks is used.

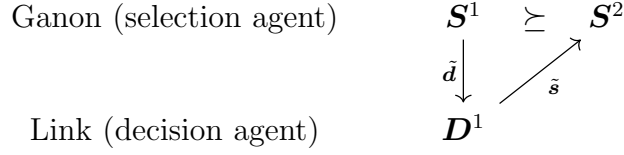
\begin{figure}[t] \centering
\begin{tikzcd}
  \text{Ganon (selection agent)} &
  \bm S^1 \arrow[d, "\tilde{\bm d}"'] \arrow[r,phantom, sloped, "\succeq"] & \bm S^2
  \\
  \text{Link (decision agent)} & \bm D^1 \arrow[ur, "\tilde{\bm s}"'] &
\end{tikzcd}
\caption{An illustration of the post-selection and extra-selection
  inference problems.}
\label{fig:one-iteration}
\end{figure}

Formally, suppose Ganon initially uses a
selection rule $\bm s^1:\mathcal{X} \to \{0,1\}^m$, Link uses a decision strategy $\tilde{\bm d}:\{0,1\}^m \times \mathcal{X} \to \mathcal{D}$,
and after observing Link's decision, Ganon applies a selection strategy
$\tilde{\bm s}: \mathcal{D} \times \mathcal{X} \to \{0,1\}^m$. Let
\[
  \bm S^1 = \bm s^1(\bm X),~\bm D^1 = \tilde{\bm d} \circ \bm s^1(\bm X)
  := \tilde{\bm d}(\bm S^1, \bm X), ~\text{and}~\bm S^2 = \tilde{\bm s} \circ
  \tilde{\bm d} \circ \bm s^1(\bm X) := \bm s(\bm D^1, \bm X), \\
\]
where in using the composition notation $\circ$ we omit the dependence
of $\tilde{\bm d}$ and $\tilde{\bm s}$ on $\bm X$. Let
$\mathcal{S}^1$ and $\mathcal{S}^2$ be
the selected sets associated with $\bm S^1$ and $\bm S^2$,
respectively, and let $\bm d^1 = \tilde{\bm d} \circ \bm s^1$. In this
context, there are two notions of
selective
risk for the decision $\bm D^1 = \bm d^1(\bm X)$: the \emph{post-selection} risk
averages the losses over the first selected set of tasks:
\[
  r^{\text{post}}(\bm d^1, \theta) = r(\bm d^1,
  \bm s^1,
  \theta) = \EE_{\theta}\left[ \frac{\sum_{i \in \mathcal{S}^1}
    \ell_i(D^1_i, \theta)}{1 \vee |\mathcal{S}^1|} \right],
\]
and the \emph{extra-selection} risk averages over the second selected set:
\[
  r^{\text{extra}}(\bm d^1, \theta) = r(\bm d^1, \tilde{\bm s}
  \circ \bm d^1, \theta) = \EE_{\theta}\left[ \frac{\sum_{i \in \mathcal{S}^2}
    \ell_i(D^1_i, \theta)}{1 \vee |\mathcal{S}^2|} \right].
\]

\begin{example}[FCR] \label{exam:FCR}
The FCR in \eqref{eq:fcr} 
is a kind of post-selection risk. Each task is about an unknown
parameter $\theta_i$. The decision space contains all intervals
$\mathcal{D}_i = \{d_i = (a_i, b_i): -\infty \le a_i \le b_i \le \infty\}$,
and the loss function is $\ell(d_i,\theta_i) = \1_{\{\theta_i \notin
  (a_i,  b_i)\}}$. The method of \textcite{benjamini2005false}
outputs a confidence interval for each selected parameter $\theta_i$,
$i \in \calS^1$ with confidence level $q|\calS^1|/\No$.
\end{example}

\begin{example}[FDR] \label{exam:FDR}
The FDR in \eqref{eq:fdr} is a kind of extra-selection risk.
Following the observation in \Cref{prop:fcr-to-fdr} and the notation
in \Cref{exam:fwer} where $\bm X = (P_1, \ldots, P_m)$, the BH procedure can be
formulated as \Cref{algo:extra.selective.risk} in \Cref{sec:constr-appr-post} with decision
strategy $\tilde{d}_i(\bm S, \bm X) = \1_{\{P_i \le q |\bm S|/m\}}$,
which corresponds to using the adjusted confidence level in
\Cref{exam:FCR}, and the ``identity'' selection strategy
$\tilde{s}_i(\bm D, \bm X) = D_i$.
\end{example}
By definition, $r^{\mathrm{post}}$
reduces to $r^{\mathrm{pre}}$ if $\mathcal{S}^1 = [m]$ (with
probability 1), i.e.\
Ganon makes no initial selection using the data.
Similarly, $r^{\mathrm{extra}}$ reduces to $r^{\mathrm{post}}$ if
$\mathcal{S}^{2} = \mathcal{S}^1$ (with
probability 1), i.e.\ Ganon makes no further
selection after observing Link's decision.

\subsection{From post-selection risk control to extra-selection risk
  control}
\label{sec:post-extra}

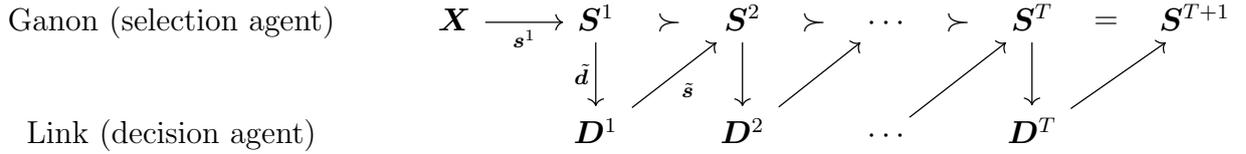
\begin{figure}[t] \centering
\begin{tikzcd}
  \text{Ganon (selection agent)} &
  \bm X
  \arrow[r, sloped, "{\bm s}^{1}"']
  &
  \bm S^1 \arrow[d,"\tilde{\bm d}"']  \arrow[r,phantom, sloped, "\succ"] &
  \bm S^2 \arrow[d]
  \arrow[r,phantom, sloped, "\succ"] & \cdots \arrow[r,phantom, sloped, "\succ"] &
  \bm S^{T} \arrow[d] \arrow[r,phantom, sloped, "="] & \bm S^{T+1}
  \\
  \text{Link (decision agent)} &
  &
  \bm D^1 \arrow[ur,"\tilde{\bm s}"'] & \bm D^2 \arrow[ur] & \cdots
  \arrow[ur] & \bm D^{T} \arrow[ur] 
  & 
\end{tikzcd}
\caption{An illustration of the iterative algorithm for
  extra-selection risk control. 
  The dependence of $\bm D^t$ and $\bm S^{t+1}$ on $\bm X$ for $t \geq
  1$ is omitted in this figure.
}
\label{fig:iterative}
\end{figure}

Consider the game between Ganon (selection agent) and Link (decision
agent) described in \Cref{sec:non-pre-post}. To generalize the
observation that the BH procedure is the iterative limit of the BY
procedure, 
we let this game continue until convergence
(\Cref{fig:iterative}). Specifically, suppose Ganon
starts this game by selecting a subset of tasks using a selection rule
$\bm s^{1}$, so $\bm S^1 = \bm s^1(\bm X)$.
In iteration $t \geq 1$, Link applies a decision strategy $\tilde{\bm d}:
\{0,1\}^m \times \calX
\to \calD$ and obtains a decision $\bm D^t = \tilde{\bm d}(\bm S^t, \bm X)$. Next,
Ganon applies a selection strategy $\tilde{\bm s}: \calD \times
\calX \to \{0,1\}^m$ and obtains a new set of tasks $\mathcal{S}^{t+1} = \{i \in
[m]: S^{t+1}_i = 1\}$ where $\bm S^{t+1} = \tilde{\bm s}(\bm D^t, \bm X)$. It is
required that Ganon can never bring back a task that he has already deselected, that is, $\mathcal{S}^{t+1} \subseteq \mathcal{S}^t$ (with
probability 1).
This game is completed the first time Ganon does not change his
selected set:
\[
  T = \inf \{t \geq 1: \calS^{t-1} = \calS^{t}\} - 1.
\]

The next result shows that taking the decision $\bm D^T$ controls the
extra-selection risk as long as the post-selection risk is controlled
at each step. To state the result, we introduce the following
notation for the cumulative selection and decision rules:
\[
  \bm s^t 
  =
  (\tilde{\bm s} \circ \tilde{\bm d})^{t-1} \circ \bm s^1,~t \geq
  2~\text{and}~
\bm d^t = 
  \tilde{\bm d} \circ \bm
                                           s^t,~t \geq 1,
\]
where in using the composition notation $\circ$, we omit the dependence
of the selection strategy $\tilde{\bm s}$ and decision strategy $\tilde{\bm d}$ on the
data. It follows that $\bm S^t = \bm s^t(\bm X)$ and $\bm D^t = \bm d^t(\bm X)$.

\begin{algorithm}[t]\caption{Extra-selection risk control}\label{algo:extra.selective.risk}
    \begin{algorithmic}
      \STATE \textbf{Input}: data $\bm X$, 
      first selection rule $\bm s^1:\calX \to \{0,1\}^m$, selection
      strategy
      $\tilde{\bm s}: \calD \times \calX \to \{0,1\}^m$, decision
      strategy $\tilde{\bm d}: \{0,1\}^m \times \calX \to \calD$.
        \STATE $t \leftarrow 1$ \\
         \STATE $\bm S^1 \leftarrow \bm s^1(\bm X)$ \\
        \REPEAT
        \STATE $\bm D^t \leftarrow \tilde{\bm d}(\bm S^t, \bm X)$ \\
        \STATE $\bm S^{t+1} \leftarrow \tilde{\bm s}(\bm D^t, \bm X)$
        \\
        \STATE $t \leftarrow t+1$ \\
        \UNTIL $\bm S^{t} = \bm S^{t-1}$
      \STATE $T \leftarrow t - 1$
      \STATE \textbf{Output}: $\bm D^T$, $\calS^{T+1}$.
    \end{algorithmic}
\end{algorithm}

\begin{theorem}\label{theo:extra.selective.risk}
  In the setting described above, consider the decision rule $\bm d^T$
  and decision rule $\bm s^{T+1}$ as defined in
  \Cref{algo:extra.selective.risk}. Suppose that 
  \begin{enumerate}
  \item [(i)] the selection rule $\tilde{\bm s} \circ \tilde{\bm d}$ is
    contracting in the sense that
    \begin{equation}
      \label{eq:contracting.selection}
      \tilde{\bm s} \circ \tilde{\bm d}(\bm S) := \tilde{\bm s}
      (\tilde{\bm d}(\bm S, \bm X), \bm X) \preceq \bm
      S~\text{with probability 1 for all}~\bm S \in \{0,1\}^m;
    \end{equation}
  \item [(ii)] the decision rule $\bm d^t$ controls the post-selection risk at level $q$ for all $t$, that is,
    \begin{equation}
      \label{eq:post.control.t}
    r^{\mathrm{post}}(\bm d^t,\theta) = r(\bm d^t, \bm
      s^t,\theta) \leq q,~\forall \theta \in \Theta,t \geq 0.
    \end{equation}
  \end{enumerate}
  Then the decision rule $\bm d^T$ controls the extra-selection risk
  at level $q$:
  \[
    r^{\mathrm{extra}}(\bm d^T,\theta) = r(\bm d^T, \bm
      s^{T+1},\theta) \leq q,~\forall \theta \in \Theta, t \geq 0.
  \]
\end{theorem}
\begin{proof}
  Because the selection is contracting, it is easy to see that $T \leq
  m + 1$ (with probability 1) and $\bm d^T = \bm d^{m+1} = \bm d^{m+2} =
  \cdots$. So the decision rule $\bm d^T$ is
  a well defined map from $\calX$ to $\calD$. By definition of
  $T$, we have $\bm S^T = \bm S^{T+1}$. Therefore,
  \[
    r^{\mathrm{extra}}(\bm d^T,\theta) =
    \EE_{\theta}\left[ \frac{\sum_{i \in \mathcal{S}^{T+1}}
        \ell_i(D^T_i, \theta)}{1 \vee |\mathcal{S}^{T+1}|} \right] =
    \EE_{\theta}\left[ \frac{\sum_{i \in \mathcal{S}^{T}}
        \ell_i(D^T_i, \theta)}{1 \vee |\mathcal{S}^{T}|} \right] =
    r^{\mathrm{post}}(\bm d^T,\theta).
  \]
  Because $T$ is random, we cannot directly use the condition (\textit{ii})
  to claim that the post-selection risk of $\bm d^T$ is less than
  $q$. However, because $\bm d^T = \bm d^{m+1}$, we have
  $r^{\mathrm{post}}(\bm d^T,\theta) = r^{\mathrm{post}}(\bm
  d^{m+1},\theta) \leq q$ as desired.
\end{proof}

In the scenario of multiple testing, \Cref{algo:extra.selective.risk} is a
constructive way to produce ``self-consistent'' selection
\parencite{blanchard2008two}. Relatedly,
\textcite{benjamini2014selective} suggested a method that can control
post-selection risk in the context of multiple testing with group
structures. They claimed without proof that one can define the
selection rule as selecting ``the largest number of families so that all
selected families will have a rejection when applying the selective
inference adjustment''
\parencite[Remark 5]{benjamini2014selective}. Since this requires the
selection rule to use post-selection information, this definition appears somewhat circular; at least it is not immediately clear how such a selection rule can be computed. In fact, this directly motivates
the iterative approach presented in this article. See
\Cref{exam:FDR.FWE} below for further discussion. 

\if\showComment1
Remark: algorithm may be understood as a fixed-point
iteration. increasing selection is also fine. If non-monotone
selection is allowed or the extended selection/decision rules are
allowed to change over the iteration, a more advanced argument (e.g.\
using martingales) may be needed.

Remark: Unifies many procedures in the literature. Relation to
self-consistency and dependency control conditions.
\zg{
    In the scenario of multiple testing, the procedure can be viewed as a constructive way to produce self-consistent selection \cite{blanchard2008two}, and a concrete recipe to ``select the largest number of families (tasks) so that all selected families will have a rejection when applying the selective inference adjustment at level $q$'' in \cite{benjamini2014selective}.
    In \cite{barber2017pfilter}, an iterative algorithm is employed to determine the set of hypotheses ensuring individual and group FDR control respecting a consistency constraint.
}
\fi

\subsection{From pre-selection risk control to post-selection risk control}
\label{sec:pre-post}

\if\showComment1
\qz{How about calling $s: \calX \to \{0,1\}^m$ a selection rule,
  $\tilde{s}: \calD \times \calX \to \{0,1\}^m$ a selection
  strategy, and $s: \{0,1\}^m \times \calX \to \{0,1\}^m$ an extended
  selection rule?}

\zg{Definitions: PosT-Selecion-Decision: PTSD
\begin{itemize}
    \item $\bm s$: selection rule
    \item $\bm s^t = \tilde {\bm s} \circ \tilde {\bm d} \circ \bm s^{t-1}$: selection rule
    \item $\bm S^t$: selection/selected set at step t
    \item $\tilde {\bm s}$: selection strategy
    \item $\bm d$: decision rule
    \item $\bm d^t = \tilde {\bm d} \circ \tilde {\bm s} \circ \bm d^{t-1}$: decision rule
    \item $\bm D^t$: decision at step t
    \item $\tilde {\bm d}$: decision strategy
\end{itemize}
}
\fi

Given decision rules that control non-selective risks, one may hope to
control the post-selection risk by adjusting the nominal level
according to the selection event. This is described in
\Cref{algo:post.selective.risk} which defines a decision rule $\bm
d^1$. Next we review two variants of the adjustment
function $f$ in this algorithm that have appeared in a series of papers
by Benjamini and coauthors.

\begin{algorithm}[t]\caption{General procedure for post-selection risk
    control}\label{algo:post.selective.risk}
    \begin{algorithmic}
         \STATE \textbf{Input}: data $\bm X$, risk level $q \in
         [0,1]$, first selection rule $\bm s^1: \calX \to
         \{0,1\}^m$, (non-selective) decision rules $d_i:\calX \times [0,1] \to
         \calD_i,i=1,\dots,m$, adjustment function $f: \NN \to \RR$. 
         \\
         \STATE $\bm S^1 \leftarrow \bm s^1(\bm X)$ \\
         \FOR{$i \leftarrow 1$ \TO $m$}
        \STATE $D_i^1 \leftarrow d_i(\bm X;
        q|\calS^1|/f(m))$
        \ENDFOR
      \STATE \textbf{Output}:
      $\calS^1$, $\bm D^1$.
    \end{algorithmic}
\end{algorithm}


The first procedure is based on the observation in
\textcite[Theorem 4]{benjamini2005false} that the FCR of the BY
procedure is always bounded by $q (m / |\mathcal{S}|)
\sum_{i=1}^m (1/i)$, regardless of the dependence structure
of the p-values. We extend this observation to
general loss functions and prove in the Supplementary Materials (see \Cref{prop:arbitrary dependence} there) that \Cref{algo:post.selective.risk} with the adjustment function $f(m) = m\sum_{i=1}^m (1/i)$ always controls the post-selection risk.


For tasks associated with independent data, a more lenient adjustment
is possible. In particular, we may use the adjustment function $f(m) =
m$, which corresponds
to the adjustment factor in the BY procedure \eqref{eq:BY}. To fix notation, suppose each task $i$ is associated with the data
$X_i$, where $X_1,\dots,X_m$ have no overlap and are independent. 
We also assume the non-selective decision rule  $d_i(\bm X; q)$ at the nominal level $q$ for task $i$ only depends on $X_i$, e.g.\ whether a p-value $P_i \le q$,
and with an abuse of notation we will 
denote it as $d_i(X_i; q)$.
 We denote $\bm X_{-i} =
(X_0,X_1,\dots,X_{i-1},X_{i+1},\dots,X_{\No})$, where $X_0$ contains
extra data that can be used in the selection and we assume $X_0
\independent (X_1,\dots,X_{\No})$.


A key notion here is stability of the selection rule.
For functions $A,
B, C$ of $\bm X$, we say $A$ is \emph{variationally independent} of $B$ given
$C$ and write $A \perp\!\!\!\!\perp_v B \mid C$ if the image of $A$
does not depend on the value of $B$ given the value of $C$; see
\textcite[Section
2.2]{constantinouExtendedConditionalIndependence2017} for a formal
definition. 

\begin{definition}[{Stable selection rule}]\label{defi:stable}
  Consider a deterministic selection rule $s: \calX \to \{0,1\}^m$ and let $\bm S = \bm s(\bm X)$ and the
  selected set $\calS $ be defined as in \Cref{sec:non-pre-post}. We
  say $s$ is \emph{stable}, if $\bm S_{-i} \perp\!\!\!\!\perp_v X_i \mid S_i = 1,
    \bm X_{-i}$ for all $i =1,\ldots,m$.
  For an extended selection rule $\bm s:\{0,1\}^m \times \calX \to \{0,1\}^m$, we say it is  stable
  if $\bm s(\bm w,\cdot)$ is  stable for every $\bm w \in
  \{0,1\}^{\No}$.
\end{definition}

Intuitively, stability means that data used for one decision does not affect whether the other
tasks are selected. 
The stability property first appeared in
\textcite{benjamini2005false} as a way to simplify a more general
FCR-controlling procedure, this was later called ``simple selection''
in \textcite{benjamini2014selective} and ``stable selection'' in
\textcite{bogomolov2018assessing}.
Most deterministic step-up, step-down selection rules are stable \parencite{benjamini2014selective}. Some of the following results only require weaker forms of stability and more details are provided in Supplementary Materials. 



The next result shows that with independent data for each task and a
stable selection rule, the more lenient adjustment function $f(m) = m$
is sufficient to control the post-selection risk. This result
is essentially the same as Theorem 1 of
\textcite{benjamini2014selective}, 
although we give a simplified and more general proof in the Supplementary Materials 
using the iterated law of expectation. 

\begin{theorem}[Post-selection risk control under independence and stability]\label{theo:post.selective.risk}
  Suppose $X_0$, $X_1$, \ldots, $X_m$ are independent, the selection
  rule
  $\bm s^1$ is stable, and $d_i$ only uses $X_i$ and controls the $\ell_i$-risk
  for all $i =1,\ldots, m$. Then \Cref{algo:post.selective.risk} with
  the adjustment function $f(m) = m$ controls
  the post-selection risk, that is,
\begin{align}\label{eq:risk.selected}
  r^{\mathrm{post}}(\bm d^1, \theta) = \EE_{\theta}\left[\frac{\sum_{i \in \calS^1}
  \ell_i\left(D_i^1,\theta\right)}{1 \vee
  |\calS^1|}\right] \le q,~\forall \theta \in \Theta~\text{and}~q \in [0,1].
\end{align}
\end{theorem}

The main new result in this section is that stability is closed under
two common operations: composition and intersection. Recall that for a
selection rule $\bm s:\calX \to \{0,1\}$ and an extended selection rule $\bm s':\{0,1\}^m \times
\calX \to \{0,1\}^m$, their composition is defined as
$\bm s'\circ \bm s (\bm X) := \bm s'(\bm s(\bm X), \bm X)$. For two selection rules $\bm s$ and $\bm s': \calX \to
\{0,1\}^m$, we define their intersection $\bm s \cap \bm s'$ as the
element-wise product of $\bm s(\bm X)$ and $\bm s'(\bm X)$, that is, a
task is selected by $\bm s \cap \bm s'$ if and only if it is selected
by both $\bm s$ and $\bm s'$.

\begin{proposition}\label{prop:composite}
For composition and intersection defined above, the following holds:
  \begin{itemize}
      \item[(1)] (Composition of stable selection rules) Let $\bm s:\calX \to \{0,1\}$ be a selection rule and $\bm s':\{0,1\}^m \times
  \calX \to \{0,1\}^m$ be an extended selection rule such that $\bm s'(\bm w; \bm X) \preceq \bm w$ for all $\bm w \in \{0, 1\}^m$. If both $\bm s$ and $\bm s'$ are stable, so is $\bm s' \circ \bm s$.
  \item[(2)]  (Intersection of stable selection
  rules) Let $\bm s$ and $\bm s':\calX \to \{0,1\}^m$ be two selection rules.
    If both $\bm s$ and $\bm s'$ are stable, so is $\bm s \cap \bm s'$.
  \end{itemize}
  
  \end{proposition}




Suppose $\tilde{\bm d}$ in \Cref{theo:extra.selective.risk} is defined as  the decision strategy in
\Cref{algo:post.selective.risk} with adjustment function $f(m) = m$, and we additionally assume
$\tilde{\bm s} \circ \tilde{\bm d}$ is stable and contracting. By the \Cref{prop:composite} (1), the composite selection rule $\bm{s}^t$ remains stable, so $\bm{d}^t$ controls the  post-selection risk by \Cref{theo:post.selective.risk}. Thus, the two conditions in \Cref{theo:extra.selective.risk} are satisfied, and \Cref{algo:extra.selective.risk} controls the extra-selection risk. 
This result is summarized in the following corollary which generalizes the FDR-controlling property of BH to compound decision problems.


\begin{corollary}[Extra-selection risk control under independence and stability]\label{coro:extra.selective.risk}
Suppose $X_0$, $X_1$, \ldots, $X_m$ are independent, the first selection rule $\bm s^1$ is stable, $d_i$ controls the $\ell_i$-risk
  for all $i =1,\ldots, m$, and $\tilde{d}_i(\bm S, \bm X)= d_i(X_i; q|\bm S|/m)$;
 suppose further
that the extended
selection rule $\tilde{\bm s} \circ \tilde{\bm d}$ is stable and contracting.
Then \Cref{algo:extra.selective.risk} controls the extra-selection risk in the sense that
\begin{align}\label{eq:extra.risk.selected}
  r^{\mathrm{extra}}(\bm d^T, \theta) = \EE_{\theta}\left[\frac{\sum_{i \in \calS^T}
  \ell_i\left(D_i^T,\theta\right)}{1 \vee
  |\calS^T|}\right] \le q,~\forall \theta \in \Theta~\text{and}~q \in [0,1].
\end{align}
\end{corollary}



\subsection{Examples}\label{sec:example}

Many existing methods in multiple testing that control FDR-like risks
can be obtained from combining \Cref{algo:extra.selective.risk}
(``post to extra'') with \Cref{algo:post.selective.risk} (``pre to
post''). We give a few examples below and include more in the supplementary file. Without further specification, we adopt $\calS^1 = [\No]$.

\begin{example}[Structural constraints]\label{exam:group.size.balance}
    Companies are increasingly turning to machine learning techniques, such as automated resume screening, semantic matching, and AI-assisted interviews, to efficiently process large pools of applicants \parencite{faliagka2012application, shehu2016adaptive}. 
    In this context, a qualified candidate is considered a true discovery, and controlling the FDR ensures that the proportion of unqualified candidates in the screened pool remains low. Moreover, promoting diversity among new hires is essential, with a focus on maintaining equitable representation across different subgroups, such as gender.

    Mathematically, we formulate this as \Cref{exam:FDR} with additional structural constraints. 
  For simple illustration, suppose each task belongs to one of two
  categories; let $X_{0i} \in \{1,2\}$ denote the category that $i$-th
  task is in. These categories could represent groups such as female and male, or different racial groups. The number of categories can be generalized to any finite number.
  For practical reasons, the rejected set
  needs to satisfy a balance constraint $n_1(\bm D)/n_2(\bm D) \in [1/\Gamma, \Gamma]$,
  where $n_j(\bm D)$ denotes the number of rejected hypotheses in $\bm
  D$ that belong to category  $j$, $j = 1, 2$, and $\Gamma > 1$ is
  some pre-specified parameter. The decision space, decision strategy, and loss function remain
  identical to those of \Cref{exam:FDR}, but we modify the selection
  strategy to incorporate the balance constraint by deselecting
  least significant p-values in the category with excessive
  rejections. 
\end{example}

\begin{example}[FDR of FWE]\label{exam:FDR.FWE}
  \textcite{benjamini2014selective} considered statistical hypotheses
  with a group structure. 
  Here, we consider a rather stringent risk---the expected
  proportion of discovered groups with any false discovery.
  Let the $i$-th statistical task, $i \in [m]$, be testing a group
  of $n_i$ null
  hypotheses, $H_{ij}:\theta \in \Theta_{ij}$, $j \in [n_i]$. Each
  task is associated with the decision space $\calD_i :=
  \{0,1\}^{n_i}$. Following the convention in \Cref{exam:fwer},
  rejecting $H_{ij}$ is denoted by $D_{ij} = 1$. Let the loss function
  for each task be the family-wise error (FWE):
  \[
    \ell_i(d_i, \theta) =
    \begin{cases}
    1, & \text{if there exists}~j \in [n_i] \text{~such that~} \theta \in
         \Theta_{ij}~\text{and}~d_{ij} = 1,\\
    0, & \text{otherwise}.
    \end{cases}
  \]
  Suppose a group is selected if one of its hypotheses is rejected,
  which corresponds to the following selection strategy
  \[
    \tilde s_i(\bm D, \bm X) = 1 - \prod_{j=1}^{n_i} (1 - D_{ij}) =
    \begin{cases}
      1, & \text{if there exists}~j \in [n_i] \text{~such that~}
           D_{ij} = 1, \\
      0, & \text{otherwise}.
    \end{cases}
  \]
  Thus, the extra-selection risk is the expectation of the average FWE
  over the selected groups, or one may say, the FDR of FWE. Let
  $\bm d_i(\cdot;q): \mathcal{X} \to \{0,1\}^{n_i}$ be any
  FWER-controlling method for task $i$ at the nominal level $q$. This can then be plugged into
  \Cref{algo:extra.selective.risk} and \Cref{algo:post.selective.risk}
  to control the FDR of FWE. This procedure iteratively applies the
  FWER-controlling method $\bm d_1,\dots,\bm d_m$ to the groups, deselects
  the groups with no rejected hypotheses, and adjusts the nominal
  level of the FWER-controlling method according to how many groups
  are left. 






\end{example}

The Supplementary Materials contains more examples that further demonstrate the flexibility of our approach to control extra-selection (FDR-like) risks.




\section{Simultaneous control of multiple selective risks}\label{sec:multiple-risks}

Next, we extend the iterative approach of selective risk control in
\Cref{sec:constr-appr-post} to the setting where we would like to
control $k$ selective risks at the same time. Specifically, for each $c
\in [\NoTask]$, suppose we have already designed  a decision strategy $\tilde{\bm d}_c$ and a selection
strategy $\tilde{\bm s}_c$
such that iterating them in \Cref{algo:extra.selective.risk} controls
the extra-selection risk defined by a loss function $\ell_c$. That is,
for some prespecified risk level $q(c), c \in [\NoTask]$, we have
\begin{align}\label{eq:multi.risk}
  r^{\mathrm{extra}}(\bm d_c,
  \theta) := \EE_{\theta}\left[\frac{\sum_{i
  \in \calS_c}
    \ell_{c,i}\left(D_{c,i},\theta\right)}{1 \vee
    |\calS_c|}\right] \le q(c)~\text{for all}~\theta \in \Theta~\text{and}~c \in [k],
\end{align}
where $\bm d_c$ is the decision rule and $\calS_c$ is the selected set
of tasks at convergence.

However,
there is no guarantee that the selected sets $\mathcal{S}_c$, $c \in
[\NoTask]$ will be the same, so the $k$ selective risks are generally
defined by averaging over different tasks. Thus, a natural question
is: is there a way to ensure ``cross-consistency'', that is the set of
selected tasks is the same for all different risks?

\begin{figure}[t] \centering
\begin{tikzpicture}[align=center,node distance=2cm]

  \node  (S1) {$\bm S^1$};
  \node at (-2, -2) (D1) {$\bm D_1^1$};
  \node at (-2,-2.5) (L1) {Link 1};
  \node at (0, -2) (D2) {$\bm D_2^1$};
  \node at (0, -2.5) (L2) {Link 2};
  \node at (2, -2) (D3) {$\bm D_3^1$};
  \node at (2, -2.5) (L3) {Link 3};

  \node at (4, -2) (S21) {$\bm S_1^2$};
  \node at (4,-2.5) (G21) {Ganon 1};
  \node at (6, -2) (S22) {$\bm S_2^2$};
  \node at (6,-2.5) (G22) {Ganon 2};
  \node at (8, -2) (S23) {$\bm S_3^2$};
  \node at (8,-2.5) (G23) {Ganon 3};
  \node at (6, 0) (S2) {$\bm S^2$};
  \node [right of=S1, xshift = 1cm] (S4) {$\succeq$};

  \draw [->] (S1) -- (D1) node[midway,left] {$\tilde{\bm d}_1$};
  \draw [->] (S1) -- (D2) node[midway,left] {$\tilde{\bm d}_2$};
  \draw [->] (S1) -- (D3) node[midway,right] {$\tilde{\bm d}_3$};
  \draw [->, bend right] (L1) to (G21);
  \node at (1,-3.8) {$\tilde{\bm s}_1$};
  \draw [->, bend right] (L2) to (G22);
  \node at (3,-3.8) {$\tilde{\bm s}_2$};
  \draw [->, bend right] (L3) to (G23);
  \node at (5.5,-3.8) {$\tilde{\bm s}_3$};
  \draw [-] (S21) -- (S2);
  \draw [-] (S22) -- (S2) node[midway,left] {$\cap$};
  \draw [-] (S23) -- (S2);

\end{tikzpicture}
\caption{An illustration of the post-selection and extra-selection
  inference problems with multiple losses via parallel intersection.}
\label{fig:multi.one-iteration}
\end{figure}
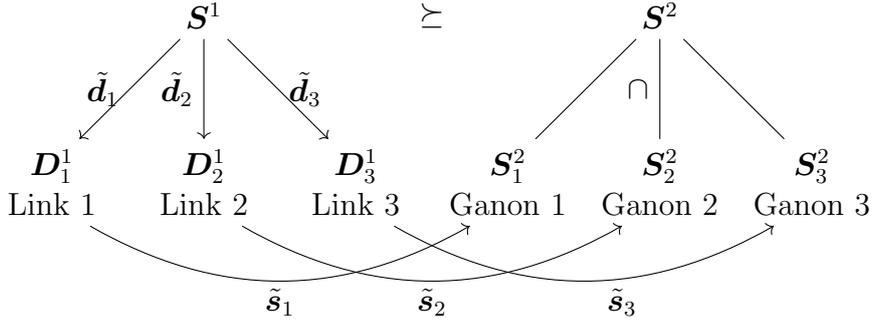

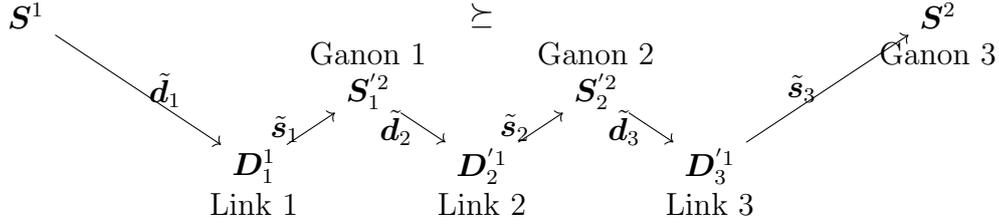
\begin{figure}[t] \centering
\begin{tikzpicture}[align=center,node distance=2cm]

  \node at (-6, 0) (S1) {$\bm S^1$};
  \node at (-3, -2) (D1) {$\bm D_1^1$};
  \node at (-1.5, -1) (S21) {$\bm S_1^{'2}$};
  \node at (0, -2) (D2) {$\bm D_2^{'1}$};
  \node at (1.5, -1) (S22) {$\bm S_2^{'2}$};
  \node at (3, -2) (D3) {$\bm D_3^{'1}$};
  \node at (6, 0) (S2) {$\bm S^2$};

  \node at (0, 0) (S4) {$\succeq$};
  \node (L1) at (-3,-2.5) {Link 1};
  \node (G21) at (-1.5,-0.5) {Ganon 1};
  \node (L2) at (0,-2.5) {Link 2};
  \node (G22) at (1.5,-0.5) {Ganon 2};
  \node (L3) at (3,-2.5) {Link 3};
  \node at (6, -0.5) {Ganon 3};

  \draw [->] (S1) -- (D1) node[midway,right] {$\tilde{\bm d}_1$};
  \draw [->] (D1) -- (S21) node[midway,left] {$\tilde{\bm s}_1$};
  \draw [->] (S21) -- (D2) node[midway,left] {$\tilde{\bm d}_2$};
  \draw [->] (D2) -- (S22) node[midway,left] {$\tilde{\bm s}_2$};
  \draw [->] (S22) -- (D3) node[midway,left] {$\tilde{\bm d}_3$};
  \draw [->] (D3) -- (S2) node[midway,left] {$\tilde{\bm s}_3$};
\end{tikzpicture}
\caption{An illustration of the post-selection and extra-selection
  inference problems with multiple losses via sequential composition.}
\label{fig:multi.one-iteration.2}
\end{figure}

\subsection{General methodology: A multi-agent game}
\label{sec:gener-meth-multi}

We show that ``cross-consistency'' can be achieved by a simple
modification to \Cref{algo:extra.selective.risk}. Intuitively, our
solution involves several decision agents, each equipped
with a decision strategy $\tilde{\bm d}_c$, and several selection
agents, each equipped with a selection strategy $\tilde{\bm s}_c$, $c
\in [k]$ in a larger game. This game can be set up in at least two
ways. In the ``parallel intersection'' setup
(\Cref{fig:multi.one-iteration}), all the decision agents
simultaneously apply their strategies to a selected set, then all
the selection agents apply their strategies to the corresponding
decisions. Finally, ``cross-consistency'' is ensured at the end of
each turn by taking the intersection of all selected
sets. Alternatively, in the ``sequential composition'' setup
(\Cref{fig:multi.one-iteration.2}), the $k$ pairs of decision and
selection agents take turns to apply their strategies, and a single
selected set is maintained. In both settings, the game is terminated
when no agent changes their decision or selection.

Next, we formally define these two procedures. For simplicity, suppose
no initial selection is made and $\mathcal{S}^1 = [m]$. The parallel
intersection procedure can be described as iterating between the
following ``aggregated'' strategies:
\begin{align}
        \tilde{\bm d}_{\text{par}}: \{0, 1\}^{\No} \times \calX &\to \prod_{c =
                         1}^{\NoTask} \calD_c,& \quad \tilde{\bm s}_{\text{par}}: \prod_{c = 1}^{\NoTask} \calD_c \times \calX \to& \{0, 1\}^{\No}, \tag*{}\\
       (\bm S', \bm X) &\mapsto \left(\tilde{\bm d}_c\left(\bm S', \bm
                         X\right)\right)_{c \in [\NoTask]},& \quad
        (\bm D_1,\dots, \bm D_{\NoTask}) \mapsto& \bigcap_{c =
                                                  1}^{\NoTask}
                                                  \tilde{\bm
                                                  s}_c\left(\bm
          D_c, \bm X \right), \label{eq:multi.aggregation}
\end{align}
See \Cref{fig:multi.one-iteration} for a demonstration. The sequential
composition procedure can be described as iterating the extended
selection rule $\bm s_{\text{seq}} = \tilde{\bm s}_k \circ
\tilde{\bm d}_k \circ \dots \circ \tilde{\bm s}_1 \circ \tilde{\bm d}_1: \{0,1\}^m
\times \calX \to \{0,1\}^m$; see \Cref{fig:multi.one-iteration.2}. The
pseudo-code for these two procedures can be found in
\Cref{app:sec:multiple-risks}.
We show in \Cref{prop:iterate.equivalence} that the two procedures generally converge to the same selected set and decisions. 

The next result extends \Cref{coro:extra.selective.risk}
(extra-selection risk control under independence and stability) to the
case of multiple risks. It almost immediately follows from the
observations that stability
is closed under composition  and intersection
(\Cref{prop:composite}). The extension to the more conservative adjustment function $f(m) = m \sum_{i=1}^m (1/i)$ (which does not require independence or stability) is straightforward.

\begin{corollary}\label{coro:multi.extra.selective.risk}
Assume the conditions in \Cref{coro:extra.selective.risk} hold for
$\bm s^1$ and all $\tilde{\bm s}_c \circ \tilde{\bm d}_c$, $c \in
[k]$. Then \Cref{algo:extra.selective.risk} with first selection
rule $\bm s^1$, decision strategy $\tilde{\bm d}_{\text{par}}$, and
selection strategy $\tilde{\bm s}_{\text{par}}$
controls all extra-selection risks with the same selected set of
tasks, that is \eqref{eq:multi.risk} is satisfied with $\mathcal{S}_c
$ being the same converged selected set for all $c$.
\end{corollary}


The general methodology outlined in this section can be applied to
obtain novel selective inference procedures that control multiple
selective risks at the same time. See \Cref{exam:partial.conjunction.map} in the Supplementary Materials for an illustration.

\subsection{Application: FDR curve}\label{sec:FDR.curve}

Next we discuss an application of the general methodology for
multiple selective risks. Continuing from \Cref{exam:fwer,exam:FDR},
let the $i$-th hypothesis be $H_{i,0}: \theta_i \geq 0$. Suppose,
however, that we are not just interested in testing $H_{i,0}$ but also
$H_{i,c}: \theta_i \geq c$ for a range of values of $c \in \calI
\subseteq \RR$. This may be useful, for example, if we are concerned
about the ``practical significance'' of a treatment, or if we simply
would like to know whether standard FDR methods (such as the
BH procedure) fail to control the selective risk when the null
hypotheses are changed slightly.

We consider the setting that the $i$-th task is also associated with an
independent test statistic $X_i \in \calX_i$ and a p-value function
indexed by $c$, $p_i: \calX_i \times \calI \to [0,1]$, such
that $P_{i,c} = p_{i}(X_i; c)$ is a valid p-value for $H_{i,c}$. This
is the case, for example, if $X_i$ is a z-score.
We can now define a FDR at every location $c \in \mathcal{I}$ and aim to
control it under a given curve $q(c)$, that is,
\begin{align*}
    \text{FDR}(c) := \EE[\text{FDP}(c)] \le q(c), \quad
    \text{FDP}(c):= \frac{\sum_{i=1}^m \1_{\{H_{i,c}  = 0\}}\1_{\{i \in
  \calS\}}}{1 \vee |\calS|}, \quad c \in \calI,
\end{align*}
where $H_{i,c} = 0$ means the null $H_{i,c}$ is true and $\mathcal{S}$
is the selected (rejection) set.

When $\mathcal{I}$ is finite, we can directly apply the parallel
intersection procedure in \Cref{sec:gener-meth-multi} to control the
FDR curve. At each iteration, we deselect task $i$ if
$
  P_{i,c} > q(c) |\mathcal{S}|/m~\text{for some}~c \in
  \mathcal{I},
$
where $\mathcal{S}$ is the current selected set. 
The method is equivalent to applying the BH procedure to the following ``adjusted p-value''
\begin{equation}\label{eq:p.value.FDR.curve}
   P_{i,\sup} = \sup_{c \in \calI} \frac{p_{i}(X_i; c)}{q(c)}, \quad i \in [\No]
 \end{equation}
 and an ``adjusted FDR level'' $q = 1$ (note that $P_{i,\sup}$ can be larger than
 $1$); see \Cref{algo:FDR.curve} in the Supplementary Materials. 
 This connection is useful when $\mathcal{I}$ is infinite and \Cref{algo:extra.selective.risk} can not be applied.
The proposition below shows that the $FDR(c)$ is controlled by $q(c)$ for $c\in \calI$ when we apply the BH procedure to $P_{i,\sup}$. In fact, this procedure provides a stronger notion of simultaneous
FDR control than \Cref{eq:multi.risk}.
\begin{proposition}\label{prop:FDR.curve}
    If $X_1,\dots,X_m$ are independent and $p_{i}(X_i; c)$ is a valid
    p-value for $H_{i,c}$ for all $i \in [\No]$, $c \in \calI$,
    then $\mathcal{S}^T$ from the BH procedure with $P_{i,\sup}$ in Eq.~\eqref{eq:p.value.FDR.curve} satisfies
    \begin{align*}
       \sup_{c \in \calI} \frac{\text{FDR}(c)}{q(c)}
        \le \EE\left[\sup_{c \in \calI}
        \frac{\sum_{i=1}^m \1_{\{H_{i,c}  = 0\}}\1_{\{i \in
      \calS^T\}}}{q(c) (1 \vee
      |\calS^T|)}\right]
        \le 1.
    \end{align*}
\end{proposition}

The standard BH procedure (applied to $P_{i,0}$, $i
\in [m]$) essentially requires that the FDR curve is controlled under
the step function $q_\text{BH}(c) := q \cdot \1_{\{c \ge 0\}} + 1
\cdot \1_{\{c < 0\}}$ for $\calI' = \{0\}$.
The next proposition shows, perhaps
surprisingly, that the BH procedure applied to $P_i = p_i(X_i; 0)$ also offers FDR control at other locations. 
\begin{proposition}\label{prop:q(c).free}
    Suppose the p-value function $p_i(x_i; c)$ is increasing in $x_i$ and decreasing in $c$ for all $i\in[\No]$.
    Then applying the BH procedure to $P_{i,0} = p_i(X_i, 0)$ at level
    $q$ produces the same rejection set as \Cref{algo:FDR.curve} with
    the following FDR curve:
    \begin{align}\label{eq:enhanced.level.curve}
        q^*(c) = 1 \wedge \left(\sup_{i \in [\No]} \sup_{x_i: q/\No \le p_{i}(x_i; 0) \le q} q \cdot \frac{p_{i}(x_i; c)}{p_{i}(x_i; 0)}\right) \le q_\text{BH}(c), \quad c \in \calI.
    \end{align}
    Thus, the BH procedure controls the FDR curve under $q^{*}(c)$.

    More generally, for arbitrary $\calI'$, the rejection set of applying \Cref{algo:FDR.curve} to control FDR$(c')$ at level $q(c')$ for $c' \in \calI'\subseteq \calI$ actually controls the FDR$(c)$ at a lower level compared to $q_\text{BH}(c) := 1 \wedge \inf_{c' \in \calI'} q(c') \cdot \1_{\{c \ge c'\}}$,
    \begin{align}\label{eq:enhanced.level.curve.2}
        q^*(c) = 1 \wedge \left(\inf_{c' \in \calI'} \sup_{i \in [\No]} \sup_{x_i: q(c')/\No \le p_{i}(x_i; c') \le q(c')} q(c') \cdot \frac{p_{i}(x_i; c)}{p_{i}(x_i; c')}\right) \le q_\text{BH}(c), \quad c \in \calI.
    \end{align}
\end{proposition}


To compute $ q^*(c) $, for~\eqref{eq:enhanced.level.curve.2}, if the p-value function is based on the CDF of a distribution with a monotone likelihood ratio, and $\calI' = \{c'\}$, the supremum over $x_i$ is guaranteed to occur at the boundary of the set of $\{x_i: q/m \le p_i(x_i;c') \le q\}$, that is the $x_i$ satisfying $q(c')/m = p_i(x_i;c')$ or $q = p_i(x_i;c')$.
If $\calI'$ contains multiple values, the problem is effectively solving multiple sub-problems with singleton $\calI'$ and then taking the infimum.

\subsection{Real data analysis: protein data} \label{sec:protein-data}

To illustrate this "free lunch" theorem, we use a real proteomic dataset \parencite{shuster2022situ} that is analyzed in \textcite{gao2023simultaneous}. The purpose of the original study was to identify candidate cell membrane proteins that influence dendrite morphogenesis in Purkinje cells, and a total of $4753$ proteins were detected \parencite{shuster2022situ}.
For each detected protein, the abundance of each protein under the treatment condition (horseradish peroxidase (HRP) + $\text{H}_2\text{O}_2$) and the control condition (HRP only) are measured.
The measurements across proteins are regarded as independent.
\textcite{shuster2022situ} are interested in plasma membrane proteins ($740$ in total), and they argue the nuclear, mitochondrial, or cytoplasmic proteins ($2067$ in total), referred to as internal negative controls in \parencite{gao2023simultaneous}, are not affected by the treatment.

To keep the notations consistent, we let $\theta_i := -(\mu_{\text{trt}} - \mu_{\text{cnt}})$ be the negative difference of the expected treatment abundance and the expected control abundance, and the hypotheses we test are $H_{i,c}: \theta_i \ge c$.
Following \textcite{gao2023simultaneous}, we obtain an empirical null distribution $\calN(\hat{\delta}, \hat{\sigma}^2)$ from the internal negative proteins, where $\hat{\delta}$ is the sample mean of the abundance differences of the internal negative proteins, and $\hat{\sigma}$ is a robust estimate of standard deviation using median absolute deviation. 
We use the difference normalized by $\hat{\delta}$ and $\hat{\sigma}$ as test statistic denoted by $X_i$, and compute the p-value for $H_{i,c}$ as $\Phi(X_i - c)$.

We apply \Cref{algo:FDR.curve} to $p(X_i; c)$ with various pre-specified $\calI'$ and $q(c)$ and compute the improved FDR level $q^*(c)$.
By the fact that the ratio $p(x; c)/p(x; c')$ is increasing in $x$ for $c > c'$ and decreasing in
$x$ for $c < c'$, $c' \in \calI'$, we obtain
\begin{align}\label{eq:FDR.curve}
    \begin{split}
        q^{*}(c)
        = 1 \wedge \inf_{c' \in \calI'} q(c')
        \cdot \left(\frac{\Phi(\Phi^{-1}(q(c')) - c)}{ \Phi(\Phi^{-1}(q(c')) - c')} \vee \frac{ \Phi(\Phi^{-1}(q(c')/\No) - c)}{\Phi(\Phi^{-1}(q(c')/\No) - c')}\right).
     \end{split}
\end{align}

We remark that $q_{\text{BH}}(c)$, $q^*(c)$ only relies on $q(c)$ and the p-value functions and is independent of the data.
The number of rejections $R$ depends on the data $X_i$.
\Cref{fig:FDR.curve.protein} shows the number of rejections $R$, as well as the comparison between $q^*(c)$ and $q_{\text{BH}}(c)$.
In (a), we adopt $\calI' = \{0\}$,
and $q^{*}(c)$ is substantially lower than $q_\text{BH}(c)$ for $c \in (-0.5,0) \cup (1, 1.5)$.
This shows that the BH procedure may offer much stronger
control of the FDR curve than $q_\text{BH}(c)$.
From (a) to (d), as we include new constraints,  $q^*(c)$ lowers significantly especially for $c \in [-2, -0.5]$, implying substantially stronger FDR$(c)$ controls.
In the protein data analysis, $R$ only decreases approximately $25\%$, indicating that the rejection set obtained from applying BH to $p_i(X_i; 0)$ (panel (a)) consists of a considerable proportion of proteins with strong signals ($\theta_i < -1.5$ or $\mu_{\text{trt}} - \mu_{\text{cnt}} > 1.5$).

\begin{figure}[tbp]
        \centering
        \begin{minipage}{0.25\textwidth}
                \centering
                \includegraphics[clip, trim = 0cm 0cm 0cm 0cm, height = 4cm]{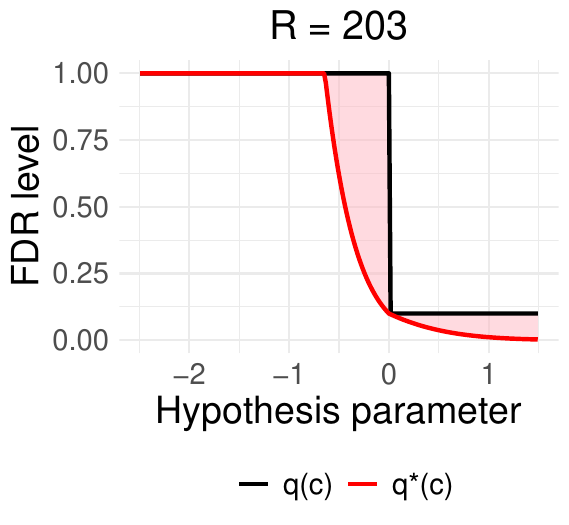}
                \subcaption{$\calI' = \{0\}$}
        \end{minipage}
        \begin{minipage}{0.23\textwidth}
                \centering
                \includegraphics[clip, trim = 0.75cm 0cm 0cm 0cm, height = 4cm]{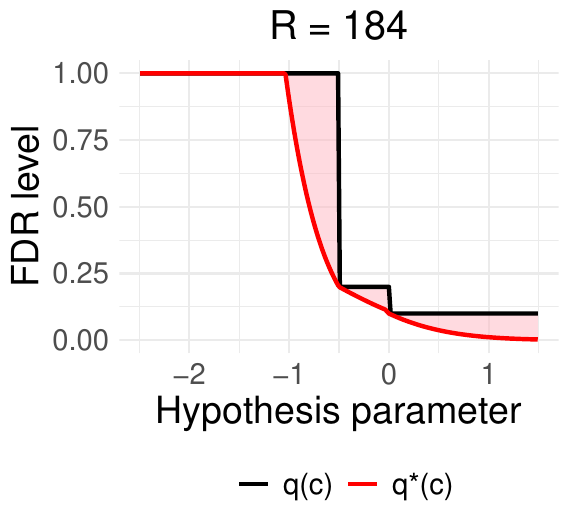}
                \subcaption{$\calI' = \{-0.5,0\}$}
        \end{minipage}
        \begin{minipage}{0.23\textwidth}
                \centering
                \includegraphics[clip, trim = 0.75cm 0cm 0cm 0cm, height = 4cm]{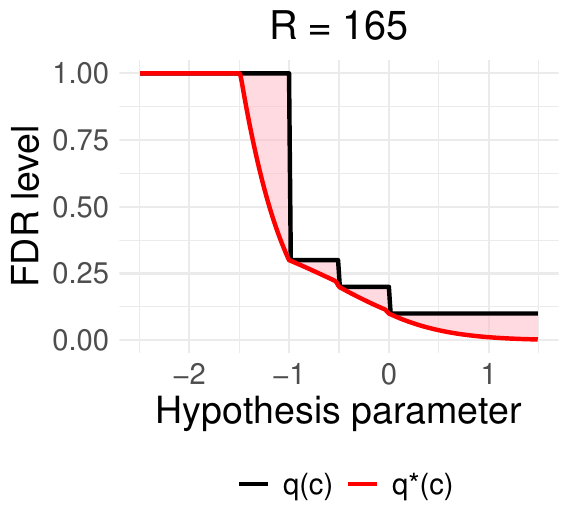}
                \subcaption{$\calI' = \{-1,-0.5,0\}$}
        \end{minipage}
        \begin{minipage}{0.26\textwidth}
                \centering
                \includegraphics[clip, trim = 0.75cm 0cm 0cm 0cm, height = 4cm]{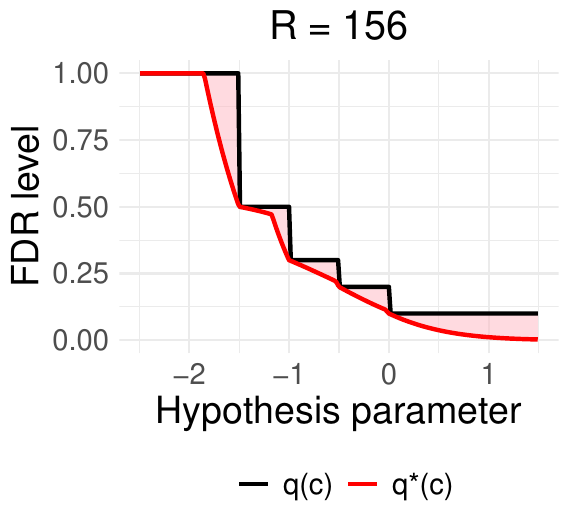}
                \subcaption{\scriptsize $\calI' = \{-1.5,-1,-0.5,0\}$}
        \end{minipage}
        \caption{{An illustration of FDR curve control using a proteomics dataset.
        We consider four $\calI'$, expanding from (a) to (d), and pre-specified FDR levels $q(c)$: $q(-1.5) = 0.5$, $q(-1) = 0.3$, $q(-0.5) = 0.2$, $q(0) = 0.1$.
        In panel (a), the rejection set is that of the BH procedure applied to the p-values for the standard null hypothesis $H_{i,0}$ (no treatment effect for protein $i$).
        }
        }
        \label{fig:FDR.curve.protein}
\end{figure}

\section{Accelerated BH with permutation tests}\label{sec:accelerate.permutation}

A permutation test assesses the significance of an observed statistic by comparing it to a reference distribution calculated from randomly permuting the original data. 
However, a major challenge in making permutation tests more widely applicable is their computational complexity, since the number of possible permutations increases exponentially in the sample size.
In this section, we show that
only a total of $O(\No \log^2\No)$ permutations and test statistic evaluations are required for combining the BH procedure with permutation tests.
The complexity $O(\No \log^2\No)$ scales only linearly in the number of hypotheses and is independent of the sample size associated with each individual hypothesis.

\subsection{Method and FDR control}

Suppose each hypothesis is associated with a test statistic $T_i$ obtained from $N_i$ data points $X_i$.
We denote the p-value exhausting all possible permutations by $P_i^\infty = {\sum_{j = 1}^{N_i!} \1_{\{T_i \le T_{ij}\}}}/{N_i!}$, where $T_{ij}$ represents the test statistic of a permuted sample.
For any $M \in \NN$, let $P_i(M)$ be a Monte Carlo approximation of the exact permutation p-value $P_i^{\infty}$ based on $M$ randomly sampled permutations,
\begin{align}\label{eq:permutation.p.value}
    P_i(M) = \frac{1 + \sum_{j = 1}^{M} \1_{\{T_i \le T_{ij}\}}}{1 + M}.
\end{align}
Under the null hypothesis, $P_i(M)$ stochastically dominates $U(0,1)$.


The standard approach
of applying the BH procedure to the approximating permutation p-values
uses the same number of permutations for each hypothesis.
In contrast, we propose an accelerated approach which, in the iteration, progressively increases the resolution of the approximating permutation p-values whose tasks are still selected.
Hypotheses with larger $P_i^{\infty}$ will be deselected in early iterations, leading to a smaller number of permutations being consumed.
The details are summarized in \Cref{algo:permutation}.

\begin{algorithm}[h]\caption{Extra-selection risk control with permutation test}\label{algo:permutation}
    \begin{algorithmic}
        \STATE \textbf{Input}: data $\bm X$, target FDR level $q \in [0,1]$, a decreasing positive integer sequence $M_r$, $r \in [\No]$.
        \STATE $t \leftarrow 1$ \\
        \STATE $\mathcal{S}^1 \leftarrow \{1,\dots,m\}$ \\
         \REPEAT
        \STATE
        For $i \in \calS^t$, generate an additional $M_{|\calS^t|} - M_{|\calS^{t-1}|}$ permutations (we set $M_{|\calS^0|} = 0$) and compute the corresponding test statistic $T_{ij}$, $M_{|\calS^{t-1}|} < j \le M_{|\calS^t|}$.
        \STATE Compute $P_i(M_{|\calS^t|})$, $i \in \calS^t$ according to Eq.~\eqref{eq:permutation.p.value} based on $T_{ij}$, $1 \le j \le M_{|\calS^t|}$.
        \STATE $\calS^{t+1} \leftarrow \{i \in \calS^{t}: P_i(M_{|\calS^t|}) \le q|\calS^{t}|/\No\}$ \\
        \STATE $t \leftarrow t+1$ \\
        \UNTIL $\mathcal{S}^{t} = \calS^{t-1}$
      \STATE $T \leftarrow t - 1$
      \STATE \textbf{Output}: $\calS^T$.
    \end{algorithmic}
\end{algorithm}

\if\showComment1
\begin{algorithm*}\caption{An iterative perspective of BH}
    \begin{algorithmic}
        \STATE \textbf{Input}: p-values $P_i$, $i \in [\No]$, FDR level $q \in [0,1]$.\\
        \STATE \textbf{Initialization}: $\calS^1 = [\No]$.\\
        \FOR{$t = 1 : \No$}
         \STATE Update $\calS^{t+1} \leftarrow \{i \in \calS^t: P_i \le q|\calS^t|/\No\}$.
         \STATE If $\calS^{t+1} = \calS^{t}$, break and output $\calS^t$.
        \ENDFOR
    \end{algorithmic}
\end{algorithm*}
\fi


\begin{proposition}\label{prop:permutation}
    If $X_i \independent \bm X_{-i}$ for $H_i = 0$ and the permutations are independently generated across all tasks,
    then \Cref{algo:permutation} controls FDR.
\end{proposition}

The expectation in the FDR of \Cref{prop:permutation} takes into account both the randomness of the observed data $\bm X$ and that of the permutations generated.
The FDR control in this result holds for arbitrary permutation number sequence $M_r$, $r \in [\No]$ in \Cref{algo:permutation}.
\Cref{prop:permutation} can be proved by adapting the argument of \Cref{theo:post.selective.risk} and details are provided in the supplementary file.


\subsection{Computation complexity}

\Cref{algo:permutation} controls the FDR for any sequence $M_r$ by \Cref{prop:permutation},
and we provide a specific recommendation in this section.
\Cref{algo:permutation} with the recommended sequence denoted by $M_{r,q,\varepsilon, \delta, \No}$ yields no less rejections compared to the BH procedure applied to $P_i^{\infty}$ at the reduced FDR level $q/(1+\delta)$ while substantially reducing the number of permutations and test statistic evaluations involved.
Let $\calS^\infty(q)$ denote the rejection set of applying BH to $P_i^\infty$ at level $q$, and $\calS(q; (M_{r})_{r \in [\No]})$ denote the rejection set of applying \Cref{algo:permutation} at level $q$ with the permutation number sequence $M_r$.

\begin{proposition}\label{prop:permutation.power}

For $\No \ge 2$, and $\varepsilon \le 0.5$, $0 < \delta \le 1$, define for $r \in [\No]$,
\begin{align}\label{defi:permutation.no}
    M_{r,q,\varepsilon, \delta, \No}
    = C_{\varepsilon, \delta, \No} \left(\frac{rq}{\No}\right)^{-1}, \quad C_{\varepsilon, \delta, \No} = 2 \left(\log\left(\frac{1}{\varepsilon}\right) + \log(\No) \right) \frac{1 + 4\delta/3 + \delta^2/3}{\delta^2}.
\end{align}
    \begin{itemize}
        \item [(a)] With probability at least $1 - \varepsilon$,
        \begin{align*}
            \calS^\infty(q/(1+\delta)) \subseteq \calS(q; (M_{r,q,\varepsilon, \delta, \No})_{k \in [\No]}).
        \end{align*}
        \item [(b)]
        \Cref{algo:permutation} with $M_{r,q,\varepsilon, \delta, \No}$ requires a maximum of
        \begin{align*}
            \frac{C_{\varepsilon, \No, \delta}}{q} \cdot \No \left(\log(\No) + 1\right)
        \end{align*}
        permutations and test statistic computations.
    \end{itemize}
\end{proposition}

We use $M_r$ to denote $M_{r,q,\varepsilon, \delta, \No}$ for notation simplicity.
By the definition~\eqref{eq:permutation.p.value} of permutation p-values, the minimum permutation p-value with $M_r$ evaluations is given by
\begin{align*}
    \frac{1}{1+M_r}
    = \frac{1}{C_{\varepsilon, \delta, \No}} \cdot \frac{rq}{\No}.
\end{align*}
When there are currently $r$ tasks selected, the rejection threshold is $rq/\No$, and at least $M_r$ (up to constant) permutations and test statistic evaluations are required to produce an approximating permutation p-value that can possibly fall below the threshold and result in the associated task remaining selected, which motivates our selection of $M_r$.

\Cref{prop:permutation.power} demonstrates that a total of $O(\No \log^2(\No))$ test statistic computations is sufficient to produce a selected set containing the oracle counterpart at a slightly reduced level $q/(1+\delta)$.
In contrast, the exact $P_i^{\infty}$ requires $\sum_{i = 1}^\No N_i!$ test statistic evaluations in total.
In \Cref{exam:permutation.no.lower.bound} of the supplementary file, we use anti-concentration bounds to prove that a slightly smaller number $O(\No \log(\No))$ of test statistic computations are necessary to recover $\calS^\infty(q/(1+\delta))$, matching our upper bound $O(\No \log^2(\No))$ up to a logarithmic term.
\cite{zhang2019adaptive} propose an algorithm which aims to recover the BH rejection set based on a pre-specified $M$ permutations.
The algorithm is shown to be of complexity $O(\No\sqrt{M})$ averaging over the generation of $P_i^{\infty}$ and the worst case $O(\No M)$.
As the authors remarked, $M$ is typically of at least the same order as $\No$, which gives the average rate $O(\No^{3/2})$ and the worst case $O(\No^2)$.


\subsection{Real data analysis: prostate cancer data}\label{sec:cancer.data}

We apply the BH procedure combined with permutation tests to the prostate cancer dataset \parencite{singh2002gene, efron2012large, efron2021computer}.
The dataset comes from a microarray study involving $102$ men, consisting of $52$ prostate cancer patients and $50$ normal controls.
In the downstream hypothesis testing, we treat the cancer status as fixed.
A total of $\No = 6033$ gene expression levels are measured for each participant.
The goal is to test for each gene the null hypothesis that the distribution of the gene's expression within the cancer group is identical to that of the normal group.

We consider three approaches: the first one using the p-values obtained from a parametric model following \textcite[Chapter 15.2]{efron2021computer}. The other two are based on permutation tests:
our proposed accelerated approach (\Cref{algo:permutation}), and the standard approach where a pre-fixed number of permutations and test statistic evaluations are carried out independently for each gene.
The FDR level is fixed at $q = 0.1$. 

The BH procedure combined with parametric p-values yields $28$ rejections.
Our proposal (hyper-parameters $\delta = 0.3$ and $\varepsilon = 0.2$) consumes an average of $894$ permutations and test statistic computations per gene and results in $29$ rejections containing all the $28$ genes discovered using the parametric p-values.
In contrast, according to \Cref{fig:permutation.number}, the standard approach with no more than $2000$ permutations per gene yields no rejections; the standard approach with $2400$ permutations per gene yields $28$ rejections, containing only $22$ out of the $28$ genes discovered using the parametric p-values. In the supplementary file, we evaluate our proposal with different hyper-parameters (\Cref{fig:permutation.number.hyperparameter}).

\begin{figure}[tbp]
        \centering
        \begin{minipage}{0.4\textwidth}
                \centering
                \includegraphics[clip, trim = 0.5cm 0cm 0cm 0cm, height = 6cm]{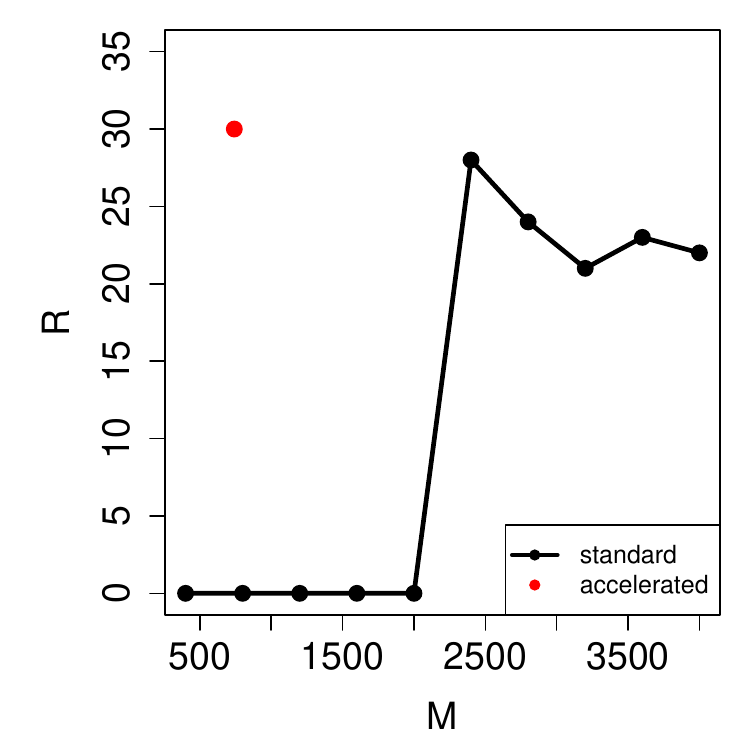}
        \end{minipage}
        \caption{
        \Cref{algo:permutation} applied to the prostate cancer data.
        We display the number of rejections obtained by the standard approach of combining the BH procedure with permutation tests with various numbers of permutations per gene $M \in [400, 800, \ldots, 4000]$. We superimpose the result (dot in the left upper corner) of \Cref{algo:permutation} ($\varepsilon = 0.2$, $\delta = 0.3$).
        }
        \label{fig:permutation.number}
\end{figure}

\section{Discussion}\label{sec:discussion}

We have developed a general and constructive approach to control a
class of selective risks. This approach allows us to unify many
existing methods, provide new insights about these methods, and
develop novel procedures for new risks in the paper. 

Stability for selection in \Cref{defi:stable} is developed for
independent p-values. \textcite{blanchard2008two} introduce a
dependency control condition that can handle a large class of
dependent distributions, but the condition is tailored for the loss
function and decision rule of the multiple testing problem. It would
be of interest to 
develop procedures that work for a wide range of loss functions and dependencies of the data across tasks.

In \Cref{sec:FDR.curve}, it is shown that the BH procedure applied
to p-values computed at a single null may provide simultaneous FDR
control at other locations, too. The FDR curve $1 - q(c)$ may be viewed as
a kind of upper bound on the cumulative distribution function of the
location parameters $\theta_i,~i\in \mathcal{S}$ for the
discoveries. This bears a close conceptual relation to the empirical
Bayes approach in \textcite{efron2012large}. To develop practical
guidance on the choice of the FDR curve $q(c)$, one can put a prior
distribution on $\theta_i$ and use the data to get an empirical
estimate of this prior distribution. Relatedly, it is well known that the actual FDR level of the BH procedure is generally lower than the nominal level, and this can be corrected using a null proportion estimator \parencite{storey2002direct,gao2024adaptive}. The same conservativeness apply to the FDR curves in \Cref{prop:q(c).free}, and it would be useful to use a similar correction.

\if\showComment1
\qz{Given that \Cref{algo:extra.selective.risk} is applicable to a wide range of loss functions, strongly selective procedures, and decision rules, it grants the flexibility to formulate user-specific strategies for tailored risk measures, such as FCR, and incorporate specific structural constraints.
\[
  R_i^{\mathrm{cond}}(d_i, \bm s, \theta) := \EE_{\theta}\left[
    \ell_i(D_i,\theta) \mid \bm S \right].
\]
Note that $R_i^{\mathrm{cond}}$ is a function of $\bm S$ and thus a
random quantity. This is much stronger than the ``average'' risk in
\eqref{eq:selective.risk}, as $R_i^{\mathrm{cond}}(d_i, \bm s, \theta)
\leq q$ for all $i$ immediately implies that $r(\bm d, \bm s,
\theta) \leq q$.
}
\fi


\if\blind0
\section*{Acknowledgement}
QZ is partly supported by the EPSRC grant EP/V049968/1.
The authors would like to thank Asaf Weinstein, Ruth Heller, and Richard Samworth for their helpful suggestions.
\fi

\printbibliography
\clearpage
\appendix

\begin{center}
    \LARGE\bf {Supplementary Materials of ``A constructive approach to selective risk control''}
\end{center}
\normalsize


This supplementary document provides additional definitions and results (\Cref{appe:sec:stable}), examples (\Cref{appe:sec:example}), and numerical results (\Cref{appe:sec:table}). Technical proofs can be found in \Cref{appe:sec:proof}.

\section{Additional definitions, algorithms, and results}\label{appe:sec:stable}

\subsection{Stable selection}

In \Cref{coro:extra.selective.risk}, it is required that the extra selection by $\tilde{\bm s} \circ \tilde{\bm
  d}$ is stable. The next result shows that this is true if the decision and selection strategies for each task are based on non-overlapping data.

\begin{proposition}\label{sufficient condition for extended stable rule}
    Suppose $\tilde{d}_i$ only depends on $\bm X$ through $X_i$ and
    $\tilde{s}_i$ only depends on $D_i$ and $X_i$. Then $\tilde{\bm
      s}\circ \tilde{\bm d}$ is stable.
\end{proposition}

We introduce some weaker definitions of stable selection rules, and extend \Cref{prop:composite}.
\begin{definition}[{Stable selection rule}]\label{defi:stable}
  Consider a deterministic selection rule $s: \calX \to \{0,1\}^m$ and let $\bm S = \bm s(\bm X)$ and the
  selected set $\calS $ be defined as in \Cref{sec:non-pre-post}. We
  say $s$ is
  \begin{enumerate}
  \item [(a)] \emph{stable}, if $\bm S_{-i} \perp\!\!\!\!\perp_v X_i \mid S_i = 1,
    \bm X_{-i}$ for all $i =1,\ldots,m$; 
  \item [(b)] 
 \emph{weakly stable}, if $|\bm S_{-i}| \perp\!\!\!\!\perp_v X_i \mid S_i = 1,
    \bm X_{-i}$ for all $i =1,\ldots,m$;
  \end{enumerate}
Furthermore, when $s$ is a stochastic function of $\bm X$, we say
  \begin{enumerate}[resume]
      \item [(c)] \emph{stochastically stable}, if $\bm S_{-i} \independent X_i \mid
    S_i = 1, \bm X_{-i}$ for all $i =1,\ldots,m$;
  \item [(d)] 
 \emph{weakly stochastically stable}, if $|\bm S_{-i}|
    \independent X_i \mid S_i = 1, \bm X_{-i}$ for all $i =1,\ldots,m$.
  \end{enumerate}
  For an extended selection rule $\bm s:\{0,1\}^m \times \calX \to \{0,1\}^m$, we say it is (weakly) stable
  if $\bm s(\bm w,\cdot)$ is (weakly) stable for every $\bm w \in
  \{0,1\}^{\No}$.
\end{definition}

The relations between these definitions can be found in \Cref{fig:stability} where $\Rightarrow$ means ``implies''.

\begin{figure}[h]
  \centering
\begin{tikzcd}
  \text{stable} \arrow[r,Rightarrow,start anchor={[xshift=5ex]},end anchor={[xshift=-5ex]}] \arrow[d,Rightarrow] &
  \text{weakly stable} \arrow[d,Rightarrow] \\
  \text{stochastically stable} \arrow[r,Rightarrow] & \text{weakly stochastically stable}
\end{tikzcd}
\caption{Relations between different notions of stable selection in \Cref{defi:stable}.}
\label{fig:stability}
\end{figure}
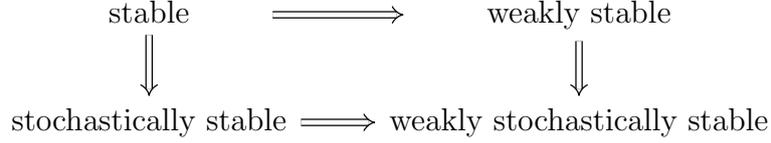

\begin{proposition}[Extension of \Cref{prop:composite}]\label{prop:composite.extended}
For composition and intersection defined above, the following holds:
  \begin{itemize}
      \item[(1)] (Composition of stable selection rules) Let $\bm s:\calX \to \{0,1\}$ be a selection rule and $\bm s':\{0,1\}^m \times
  \calX \to \{0,1\}^m$ be an extended selection rule such that $\bm s'(\bm w; \bm X) \preceq \bm w$ for all $\bm w \in \{0, 1\}^m$. If $\bm s'$ is (weakly) stable and  $\bm s$ is stable, then $\bm s' \circ \bm s$ is (weakly) stable.
  \item[(2)]  (Intersection of stable selection
  rules) Let $\bm s$ and $\bm s':\calX \to \{0,1\}^m$ be two selection rules.
    If $\bm s'$ is (stochastically) stable, and $\bm s$ is (stochastically) stable, then $\bm s \cap \bm s'$ is (stochastically) stable.
  \end{itemize}
  
  \end{proposition}

\subsection{Adjustment function for dependent data}

As promised in \Cref{sec:pre-post}, we state a result that gives post-selective risk control with arbitrarily dependent data using the more conservative adjustment function in \textcite{benjamini2001control}; see also \textcite[Theorem 4]{benjamini2005false}.

\begin{theorem}[]\label{prop:arbitrary dependence}
    Suppose
\begin{itemize}
    \item[(i)] $d_i$ controls the $\ell_i$-risk, i.e.\ $r_i(d_{i}(\cdot; q), \theta)
    \le q$, $\forall~q \in [0,1]$;
    \item[(ii)] $\ell_i\left(d_i(\cdot; q), \theta \right)$ is non-decreasing in $q$.
\end{itemize}
Then \Cref{algo:post.selective.risk} with the adjustment function $f(m)
= m \sum_{i=1}^{m} (1/i)$ controls the post-selection risk, that is,
\begin{align}\label{eq:risk.selected}
    r^{\text{post}}(\bm d^1, \theta)
    = \EE_{\theta}\left[\frac{\sum_{i \in \calS^1} \ell_i\left(D_i^1, \theta\right)}{1 \vee |\calS^1|}\right]
    \le q, ~\forall~\theta \in \Theta ~\text{and}~ q \in [0,1].
\end{align}
\end{theorem}

\subsection{Multiple selective risks}\label{appe:sec:multiple}

\begin{algorithm}[H]
  \caption{Parallel intersection method for multiple
    extra-selection risk control}\label{algo:multiple extra.selective.risk.parallel}
    \begin{algorithmic}
         \STATE \textbf{Input}: data $\bm X$, first selection rule $\bm s^1:\calX \to \{0,1\}^m$,  
         selection strategies $\tilde{\bm s}_c: \calD \times \calX \to \{0,1\}^m$, decision strategies $\tilde{\bm d_c}:  \{0,1\}^m \times \calX \to \calD$, $c \in [k]$.\\

                 \STATE $t \leftarrow 1$ \\
         \STATE $\bm S^1 \leftarrow \bm s^1(\bm X)$ \\
        \REPEAT
        \FOR{$c \leftarrow 1$ \TO $k$}
        \STATE $\bm D_c^t \leftarrow \tilde{\bm d}_c(\bm S^{t}, \bm X)$ \\
        \STATE $\bm S_{c}^{t+1} \leftarrow \tilde{\bm s}_c(\bm D_c^t, \bm X)$
        \\
        \ENDFOR
        \STATE $\bm S^{t+1} \leftarrow \cap_{c \in [k]} \bm S_c^{t+1}$        
        \STATE $t \leftarrow t+1$ \\
        \UNTIL $\bm S^{t} = \bm S^{t-1}$
      \STATE $T \leftarrow t - 1$
      \STATE \textbf{Output}: $\bm D^T$, $\calS^{T+1}$.
    \end{algorithmic}
\end{algorithm}

\begin{algorithm}[H]
  \caption{Sequential composition method for multiple
    extra-selection risk control.
}\label{algo:multiple extra.selective.risk}
    \begin{algorithmic}
         \STATE \textbf{Input}: data $\bm X$, first selection rule $\bm s^1:\calX \to \{0,1\}^m$,  
         selection strategies $\tilde{\bm s}_c: \calD \times \calX \to \{0,1\}^m$, decision strategies $\tilde{\bm d_c}:  \{0,1\}^m \times \calX \to \calD$, $c \in [k]$.\\

        \STATE $t \leftarrow 1$ \\
         \STATE $\bm S^1 \leftarrow \bm s^1(\bm X)$ \\
        \REPEAT
        \STATE $\bm S^{'t+1}_{0} \leftarrow \bm S^t$
        \FOR{$c \leftarrow 1$ \TO $k$}
        \STATE $\bm D_c^t \leftarrow \tilde{\bm d}_c(\bm S_{c-1}^{'t+1}, \bm X)$ \\
        \STATE $\bm S_{c}^{'t+1} \leftarrow \tilde{\bm s}_c(\bm D_c^t, \bm X)$
        \\
        \ENDFOR
        \STATE $\bm S^{t+1} \leftarrow \bm S_k^{'t+1}$        
        \STATE $t \leftarrow t+1$ \\
        \UNTIL $\bm S^{t} = \bm S^{t-1}$
      \STATE $T \leftarrow t - 1$
      \STATE \textbf{Output}: $\bm D^T$, $\calS^{T+1}$.
    \end{algorithmic}
\end{algorithm}

The next result shows that the parallel intersection procedure (\Cref{algo:multiple extra.selective.risk.parallel}) and the sequential composition procedure (\Cref{algo:multiple extra.selective.risk}) generally converge to the same selected set and decisions. 

\begin{proposition}\label{prop:iterate.equivalence}
Suppose for every $c \in [k]$, the extended selection rule $\tilde{\bm
  s}_c\circ \tilde{\bm d}_c:\{0,1\}^m \times \calX \to \{0,1\}^m$ is
contracting in the sense of \eqref{eq:contracting.selection} and
increasing in the sense that
    \begin{equation}
      \label{eq:increasing.selection}
      \tilde{\bm s}_c \circ \tilde{\bm d}_c(\bm S) \preceq \tilde{\bm s}_c
      \circ \tilde{\bm d}_c(\bm S') ~\text{with probability 1 for
        all}~\bm S, \bm S' \in \{0,1\}^m~\text{such that}~\bm S \preceq
      \bm S'.
    \end{equation}
Then the parallel intersection and sequential composition procedures
converge to the same selected set and decisions.
\end{proposition}

It can be verified that all the selection and decision strategies in
the examples in \Cref{sec:example} satisfy the increasing property
\eqref{eq:increasing.selection}.

As mentioned in \Cref{sec:FDR.curve}, one can modify the BH procedure to control FDR at multiple locations (or a whole FDR curve) using the adjusted p-values in \eqref{eq:p.value.FDR.curve}. Pseudo-code for this modified BH procedure can be found in \Cref{algo:FDR.curve}.

\begin{algorithm}[H]
\caption{Modified BH procedure for controlling the
    FDR curve}\label{algo:FDR.curve}
    \begin{algorithmic}
    \STATE \textbf{Input}: test statistics $\bm X$, p-value function
    $p_i:\mathcal{X}_i \times \mathcal{I} \to [0,1]$, target FDR curve
    $q:\mathcal{I} \to
    \mathbb{R}_{>0}$. \\
    \FOR{$i \leftarrow 1$ \TO $m$}
      \STATE Compute $P_{i,\sup}$ using \eqref{eq:p.value.FDR.curve}.
    \ENDFOR
    \STATE $t \leftarrow 1$ \\
    \STATE $\mathcal{S}^1 \leftarrow \{1,\dots,m\}$ \\
    \REPEAT
        \STATE $\calS^{t+1} \leftarrow \{i \in \calS^{t}: P_{i,\sup} \le |\calS^{t}|/\No\}$ \\
        \STATE $t \leftarrow t+1$ \\
        \UNTIL $\mathcal{S}^{t} = \calS^{t-1}$
      \STATE $T \leftarrow t - 1$
      \STATE \textbf{Output}: $\calS^T$.
    \end{algorithmic}
\end{algorithm}

\section{Additional examples}\label{appe:sec:example}

\begin{example}[FDR control under arbitrary dependence]
This example iterates the decision strategy $\tilde{\bm d}$ as
defined in the for loop in \Cref{algo:post.selective.risk}, which is
motivated by \textcite[Theorem 1.3]{benjamini2001control}. The fact
that this procedure controls the FDR under arbitrary dependence is a
corollary of \Cref{theo:extra.selective.risk} and \Cref{prop:arbitrary
  dependence}. A direct proof of this result was given by
\textcite[p.\ 1182-1183]{benjamini2001control}.
\end{example}

Now we consider examples based on \Cref{coro:extra.selective.risk} for independent individual data,
which requires that the extra selection by $\tilde{\bm s} \circ \tilde{\bm
  d}$ is stable. \Cref{sufficient condition for extended stable rule} provides a useful result that shows this is true if the
decision and selection strategies for each task are based on
non-overlapping data.

\begin{example}[FDR] \label{exam:FDR-proof}
  Continuing from \Cref{exam:FCR} and \Cref{exam:FDR}, the
  FDR-controlling property of the BH procedure with independent
  p-values is a corollary of \Cref{theo:extra.selective.risk},
  \Cref{coro:extra.selective.risk}, and \Cref{sufficient condition for
    extended stable rule}. As noted by
  \textcite{benjamini2005false}, this can further be strengthened to
  control the directional FDR. This is investigated in
  \Cref{exam:directional-fdr} in the Appendix, and the risk controlling
  property almost immediately follows from writing directional FDR as
  an extra-selection risk with the $i$-th decision space being
  $\mathcal{D}_i = \{+, -, \mathrm{nil}\}$ where $\mathrm{nil}$ means
  undecided sign.
\end{example}

\begin{example}[Directional FDR] \label{exam:directional-fdr}

In the context of drug experiments, researchers aim to assess whether the new drug demonstrates effectiveness to unit $i$, denoted as $\theta_i > 0$, or harmful effects, denoted as $\theta_i < 0$.
Using our notations, each task $i$ consists of two directional hypotheses $H_{i}^{-}: \theta \in  \Theta_i^- := \{\theta_i \le 0\}$ and $H_{i}^{+}: \theta \in  \Theta_i^+ := \{\theta_i \ge 0\}$.
The associated decision space contains three elements $\calD_i := \{-,+,\emptyset\}$, where $D_i = +$ indicates a positive $\theta_i$, $D_i = -$ indicates a negative $\theta_i$, and $D_i = \emptyset$ means that the sign of $\theta_i$ is undetermined.
  The loss function is defined as
  \[
    \ell_i(d_i, \theta) =
    \begin{cases}
    1, & \text{if}~\theta \in \Theta_i^-~\text{and}~d_i = + ~\text{or}~ \theta \in \Theta_i^+~\text{and}~d_i = -,\\
    0, & \text{otherwise}.
    \end{cases}
  \]
  The extra-selection risk is the expected fraction of parameters with inaccurately identified signs among those whose signs are declared.
  In particular, we have data $\bm p^-= (p_1^-,\ldots, p_m^-)$, where $P_i^-$ is a p-value for $H_i^-$, and we adopt the decision strategy $\tilde d_i(P_i^-)$ equals $+$ if $1 - P_i^- \le q$, $-$ if $P_i^- \le q$, and $\emptyset$ otherwise, where $q < 0.5$. And $\tilde s_i(\bm D, \bm p^-) = 1$ if and only if $D_i \neq \emptyset$.
  Then \Cref{algo:extra.selective.risk} recovers the directional BH procedure in \parencite{benjamini2005false}.
\end{example}

\begin{example}[Partial conjunction map]\label{exam:partial.conjunction.map}
    We continue with the group structure discussed in
    \Cref{exam:FDR.FWE}. But instead of rejecting individual
    hypotheses, we are satisfied with just obtaining some evidence of
    replicability for the selected groups. One way to formalize this
    is to provide a confidence lower bound for the total number of
    non-nulls, denoted by $n_{1,i}(\theta)$, if group $i$ is
    selected. Let the decision space and loss function be defined,
    respectively, as
    \[
      \calD_i = [n_i] \cup \{0\} \quad \text{and} \quad
      \ell_i(d_i, \theta) = \1_{\{d_i > n_{1,i}(\theta)\}}.
    \]
    Let $P^{(k)}_i$ be a p-value for the partial conjunction null
    hypothesis $H_{i}^{(k)}: n_{1,i}(\theta) < k$. This is usually
    obtained by taking the maximum of all the p-values for testing the
    intersection of any $k$ hypotheses in $H_{ij}, j \in [n_i]$. For
    simplicity, let $P^{(0)}_i = 0$. Iterating the decision strategy
    $\tilde d_i(\bm S, \bm X) = \sup\{k \geq 0: P_{i}^{(k)} \le q |\bm
    S| / m\}$ and selection
    strategy $\tilde s_i(\bm D, \bm X) = \1_{\{D_i \geq 1\}}$, $i \in [m]$ in
    \Cref{algo:extra.selective.risk} recovers the procedure of 
    \textcite[Section 4]{benjamini2008screening}, who called the
    corresponding extra-selection risk ``the FDR for screening at all
    levels''.
\end{example}

Finally, we consider an example controlling multiple selection risks.
\begin{example}[FDR of FWE and FCR of confidence bounds] \label{exam:multiple.risk.confidence.bound}
Consider the multiple testing problem in \Cref{exam:FDR.FWE} 
with a group structure. By
applying the parallel intersection procedure to the decision and
selection strategies defined in those examples, we can obtain a
selected set $\mathcal{S}$ and a lower confidence bound $D_i$ for the
number of non-nulls for each $i \in \mathcal{S}$ such that the
FDR of FWE (a selective risk for $\mathcal{S}$) and the FCR of the
confidence bounds (a selective risk for $d_i,i\in\mathcal{S}$) are
both controlled at desired levels. This is generally not possible if
one simply takes the intersection of the selected sets given by the
procedures in \Cref{exam:FDR.FWE,exam:partial.conjunction.map}.
\end{example}

\section{Additional tables and figures}\label{appe:sec:table}

\begin{table}[h]
\centering
\caption{Data of \Cref{fig:iterate.FCR}. 
The q-value is defined as $Q_i = P_i / (r_i / \No)$, $\No = 20$, where $r_i$ denotes the rank of $P_i$. The BH procedure finds the largest $P_i$ such that $Q_i \le q = 0.3$ and rejects all hypotheses with smaller or equal p-values. We use $U_i^t$ to denote the upper boundary of the one-sided confidence interval for parameter $\theta_i$ in the $t$-th iteration.
}\label{tab:iterative.FCR}
\begin{tabular}{c|c|ccc|c|c}
\toprule
\multirow{2}{*}{Rank} & \multirow{2}{*}{$X_i$} & \multicolumn{3}{c|}{$U_i^t$}                                             & \multirow{2}{*}{$P_i$} & \multirow{2}{*}{$Q_i$} \\
                  &                     & \multicolumn{1}{c}{$t = 1$} & \multicolumn{1}{c}{$t = 2$} & \multicolumn{1}{c|}{$t = 3$} &                     &                     \\ \midrule
1                 & -2.59               & -2.07                & -1.42                & -1.25                & 0.00473             & 0.0947              \\
2                 & -2.16               & -1.63                & -0.982               & -0.816               & 0.0155              & 0.155               \\
3                 & -2.14               & -1.62                & -0.968               & -0.803               & 0.016               & 0.107               \\
4                 & -2.02               & -1.49                & -0.84                & -0.674               & 0.0219              & 0.11                \\
5                 & -1.88               & -1.35                & -0.703               & -0.537               & 0.0302              & 0.121               \\
6                 & -1.68               & -1.15                & -0.504               & -0.339               & 0.0465              & 0.155               \\ \midrule
7                 & -1.1                & -0.575               & 0.0754               & 0.241                & 0.136               & 0.388               \\
8                 & -0.755              & -0.231               & 0.42                 & 0.586                & 0.225               & 0.563               \\ \midrule
9                 & -0.158              & 0.366                & 1.02                 & 1.18                 & 0.437               & 0.971               \\
10                & -0.136              & 0.388                & 1.04                 & 1.2                  & 0.446               & 0.892               \\
11                & -0.0408             & 0.484                & 1.13                 & 1.3                  & 0.484               & 0.88                \\
12                & -0.0293             & 0.495                & 1.15                 & 1.31                 & 0.488               & 0.814               \\
13                & 0.167               & 0.692                & 1.34                 & 1.51                 & 0.566               & 0.871               \\
14                & 0.245               & 0.769                & 1.42                 & 1.59                 & 0.597               & 0.852               \\
15                & 0.499               & 1.02                 & 1.67                 & 1.84                 & 0.691               & 0.921               \\
16                & 0.702               & 1.23                 & 1.88                 & 2.04                 & 0.759               & 0.948               \\
17                & 0.755               & 1.28                 & 1.93                 & 2.1                  & 0.775               & 0.911               \\
18                & 0.779               & 1.3                  & 1.95                 & 2.12                 & 0.782               & 0.869               \\
19                & 1.01                & 1.53                 & 2.19                 & 2.35                 & 0.844               & 0.888               \\
20                & 1.88                & 2.4                  & 3.05                 & 3.22                 & 0.97                & 0.97          \\ \bottomrule
\end{tabular}
\end{table}

\Cref{tab:iterative.FCR} contains the data for \Cref{fig:iterate.FCR} and additional illustration of BH as the fixed-point iteration of BY.


\Cref{fig:permutation.number.hyperparameter} evaluates \Cref{algo:permutation} in \Cref{sec:accelerate.permutation} on the prostate cancer data. 
We assess the reliability of our proposal \Cref{algo:permutation} with different hyper-parameters:
\begin{itemize}
    \item [(a)] $\varepsilon \in \{0.05, 0.1, 0.15, 0.2, 0.25\}$, $\delta = 0.3$;
    \item [(b)] $\delta \in \{0.1, 0.2, 0.3, 0.4, 0.5\}$, $\varepsilon = 0.2$.
\end{itemize}
\Cref{fig:permutation.number.hyperparameter} (a), (b) show the average number of permutations consumed increases only very slowly as $\varepsilon$ decreases, but is more sensitive to $\delta$.

\begin{figure}[t]
        \centering
        \begin{minipage}{0.45\textwidth}
                \centering
                \includegraphics[clip, trim = 0.5cm 0cm 0cm 0cm, height = 6cm]{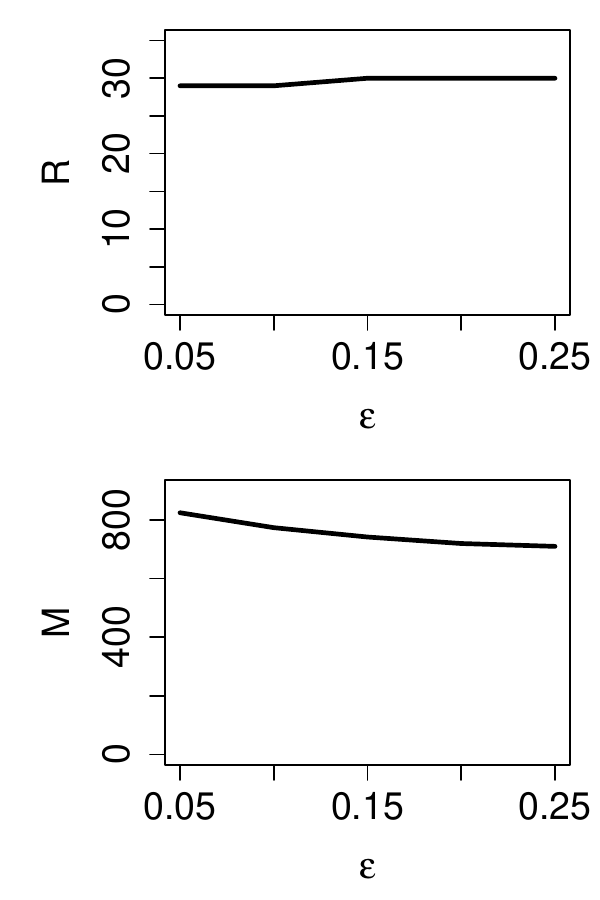}
                \subcaption*{\quad\quad (a)}
        \end{minipage}
        \begin{minipage}{0.45\textwidth}
                \centering
                \includegraphics[clip, trim = 0.5cm 0cm 0cm 0cm, height = 6cm]{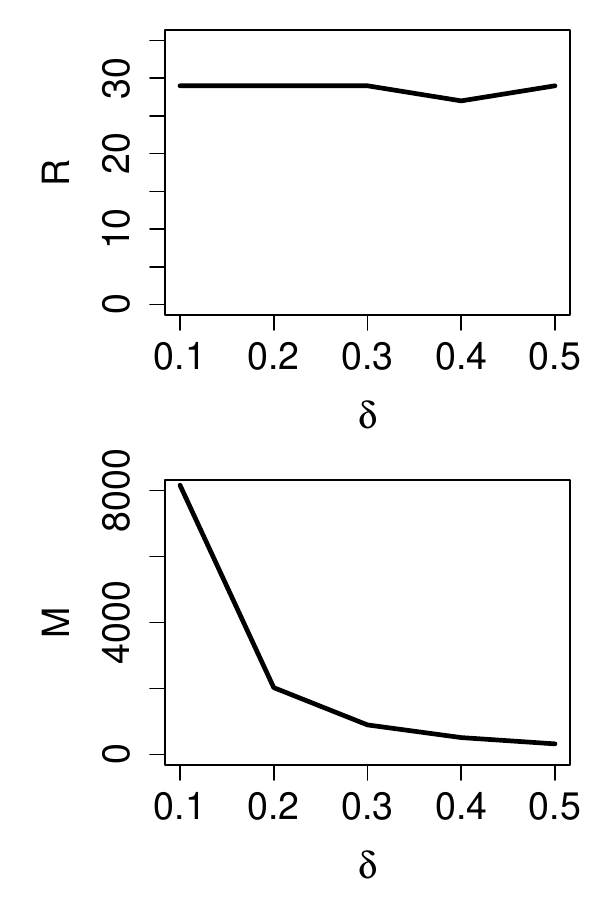}
                \subcaption*{\quad\quad (b)}
        \end{minipage}
        \caption{
        Evaluation of \Cref{algo:permutation} on the prostate cancer data.
        We adopt the data and methods in \Cref{sec:cancer.data}.
        The default hyper-parameters are $\varepsilon = 0.2$, $\delta = 0.3$.
        In panel (a) and (b), we evaluate the stability of the number of rejections produced by \Cref{algo:permutation} and the average number of permutations consumed regarding the hyper-parameters $\varepsilon$, $\delta$, respectively.
        }
        \label{fig:permutation.number.hyperparameter}
\end{figure}

\section{Technical proofs}\label{appe:sec:proof}

\subsection{\Cref{sec:intro}}

\begin{proof}[Proof of \Cref{prop:fcr-to-fdr}]
    Since $|\calS^{1}| < \infty$ and that $|\calS^{t+1}| \neq
    |\calS^{t}|$ implies $|\calS^{t+1}| \le |\calS^{t}| - 1$, the
    algorithm converges in finite time.
    Next, it suffices to show that  $|\calS^T| = I^*$ for $I^*$ in \eqref{eq:bh}.
    On one hand, if $|\calS^t| \ge I^*$, then
    \begin{equation} \label{eq:bh-iteration}
        \calS^{t+1}
        = \{i:  U_i^t \le 0\}
        = \{i: P_i \le q|\calS^{t}|/\No\}
        \supseteq \{i: P_i \le qI^*/\No\}.
    \end{equation}
    This implies $|\calS^{t+1}| \ge I^*$, and eventually $|\calS^{T}| \ge I^*$. On the other hand,
    \begin{align*}
        \calS^{T+1} = \{i: P_i \le q|\calS^{T}|/\No\}
        = \calS^{T}
    \end{align*}
    indicates that $P_{(|\calS^{T}|)} \le q|\calS^{T}|/\No$.
    By the maximality of $I^*$ in \eqref{eq:bh}, we obtain $I^* \ge |\calS^{T}|$.
\end{proof}

\subsection{\Cref{sec:constr-appr-post}}

\begin{proof}[Proof of \Cref{theo:post.selective.risk}]

In fact, this theorem also holds when $\bm{s}$ is  weakly stochastically stable. So we will prove this theorem under weakly stochastic stability.

We denote the conditional distribution of $|\calS^1| \mid \bm X_{-i}, S^1_i = 1$ as $\PP_{i}(\cdot; \bm X_{-i})$.
Let $Z_i$ be a random variable such that $Z_i \independent (X_i, S_i^1)\mid \bm X_{-i}$  and $Z_i\mid \bm X_{-i}$ follows the distribution $\PP_{i}$.
Then we define ${|\calS^1|}^{(-i)} = |\calS^1|$ if $S^1_i = 1$ and  ${|\calS^1|}^{(-i)} =Z_i$ if  $S^1_i = 0$.
By \Cref{lemm:independence} below, we have ${|\calS^1|}^{(-i)} \independent X_i\mid \bm X_{-i}$.
Together with $X_i \independent \bm X_{-i}$, we have $X_i \independent {|\calS^1|}^{(-i)}$ by the contraction rule of conditional independence.
For simplicity, let $\tilde{q} := q/m$.
For task $i$,
    \begin{align*}
    \EE\left[\frac{S^1_i \ell_i( d_i(\bm S^1, X_i;|\calS^1|\tilde{q} ), \theta)}{1 \vee |\calS^1|}\right] 
        &= \EE\left[\frac{S^1_i \ell_i\left( d_i(\bm S^1, X_i;|\calS^1|^{(-i)}\tilde{q} ), \theta\right)}{1 \vee {|\calS^1|}^{(-i)}}\right] 
        \tag*{$(\text{Definition of}~|\calS^1|^{(-i)})$}
        \\
        &\le \EE\left[\frac{ \ell_i\left( d_i(\bm S^1, X_i;|\calS^1|^{(-i)}\tilde{q} ),\theta\right)}{1 \vee {|\calS^1|}^{(-i)}}\right] 
        \tag*{$(S^1_i \le 1, \ell_i \ge 0)$}
        \\   &=\EE\left[\EE\left[\frac{\ell_i\left( d_i(\bm S^1, X_i;|\calS^1|^{(-i)}\tilde{q} ),\theta\right)}{1 \vee |\calS^1|^{(-i)}} \mid |\calS^1|^{(-i)}\right] \right] 
        \tag*{$(\text{Tower property})$}
        \\
        &=\EE\left[\frac{\EE\left[\ell_i\left( d_i(\bm S^1, X_i;|\calS^1|^{(-i)}\tilde{q} ),\theta\right)\mid |\calS^1|^{(-i)}\right]}{1 \vee |\calS^1|^{(-i)}}  \right]\\
        &\le \EE\left[\frac{|\calS^1|^{(-i)} \tilde{q}}{1 \vee |\calS^1|^{(-i)}}\right] 
        \tag*{$(|{\calS^1}|^{(-i)} \independent X_i, d_i~\text{controls $\ell_i$-risk})$}  
        \\
        &\le \tilde{q}.
    \end{align*}
Sum over $i \in [\No]$, and we have finished the proof of post-selection risk control.
\end{proof}

\begin{lemma}\label{lemm:independence}
    ${|\calS^1|}^{(-i)} \independent X_i \mid \bm X_{-i}$.
\end{lemma}

\begin{proof}[Proof of \Cref{lemm:independence}]

For simplicity, we use $\calS$ to denote $\calS^1$. For $k \in \{1, \ldots,m\}$,
\begin{align*}
    &\quad~\PP\left({|\calS|}^{(-i)} = k \mid X_i,  \bm X_{-i} \right)\\
    &=\PP\left({|\calS|}^{(-i)} = k \mid S_i = 1, X_i, \bm X_{-i} \right) \PP(S_i = 1 \mid X_i,  \bm X_{-i} ) \\
    &\quad~+ \PP\left({|\calS|}^{(-i)} = k \mid S_i = 0, X_i,  \bm X_{-i} \right) \PP(S_i = 0 \mid \bm
 X_i,  \bm X_{-i} ) 
 \tag*{$(\text{Tower property})$}
 \\
    &=\PP\left(|\calS| = k \mid S_i = 1, X_i,  \bm X_{-i} \right) \PP(S_i = 1 \mid X_i,  \bm X_{-i}) \\
    &\quad~+ \PP\left(Z_i = k \mid S_i = 0, X_i,  \bm X_{-i} \right) \PP(S_i = 0 \mid X_i,  \bm X_{-i}) 
    \tag*{$(\text{Definition of}~{|\calS|}^{(-i)})$}
    \\
    &=\PP\left(|\calS| = k \mid S_i = 1, \bm X_{-i} \right) \PP(S_i = 1 \mid X_i, \bm X_{-i}) \\
    &\quad~+ \PP\left(Z_i = k \mid \bm X_{-i}\right) \PP(S_i = 0 \mid X_i, \bm X_{-i}) 
    \tag*{$(|\calS| \independent X_i \mid S_i = 1, \bm X_{-i}; Z_i \independent (X_i, S_i) \mid \bm X_{-i})$}
    \\
    &= \PP_i(k ; \bm X_{-i}) \left(\PP(S_i = 1 \mid X_i, \bm X_{-i}) + \PP(S_i = 0 \mid X_i, \bm X_{-i})\right)
    \quad 
    \tag*{($Z_i \mid \bm X_{-i} ~\text{and}~|\calS|\mid \bm X_{-i}, S_i = 1 \sim \PP_i(\cdot; \bm X_{-i})$)} 
    \\
    &= \PP_i(k ; \bm X_{-i}).
\end{align*}
Therefore, $\PP\left({|\calS|}^{(-i)} = k \mid X_i,  \bm X_{-i} \right)$ does not depende on $X_i$, which is equivalent to ${|\calS|}^{(-i)} \independent X_i | \bm X_{-i}$.
\end{proof}

\begin{proof}[Proof of \Cref{exam:group.size.balance}]

  Although the conditions in \Cref{sufficient condition for extended stable rule} are not satisfied in this example, it can be shown that the implied extended selection rule is still stable, as proved below.
  
  For simplicity, we assume there are no ties among the p-values. The
  modified selection strategy is then defined as
  \[
    \tilde{s}_i(\bm D, \bm X) =
    \begin{cases}
      1, & \text{if}~D_i = 1~\text{and}~\sum_{j=1}^m \1_{\{X_{0j} =
           X_{0i}\}} \1_{\{P_i \ge P_j\}} \leq \Gamma
           \min \{n_1(\bm D), n_2(\bm D)\}, \\
      0, & \text{otherwise}.
    \end{cases}
  \]
  It is easy to see that this extra selection meets the balance
  constraint. We now show that the resulting extended selection rule
  $\tilde{\bm s} \circ \tilde{\bm d}$ is stable.
    Conditioning on $\bm X_{-i}$ and $S_i = 1$, we know $D_i = 1$ and $D_j$ is fixed for $j \neq i$.
    Consequently, $n_1(\bm D)$, $n_2(\bm D)$ are fixed.
    Without loss of generality, suppose $n_1(\bm D) \ge n_2(\bm D)$.
    If $i$ belongs to category $2$, the selected hypotheses will not
    change because $\tilde{\bm s}$ only removes hypotheses in category
    $1$. If $i$ belongs to category $1$, we denote the order
    statistics of the rejected p-values in category $1$ besides $P_i$ as
    $\tilde{P}_1 < \dots < \tilde{P}_{n_1(\bm D) - 1}$.
    The modified selection strategy $\tilde{\bm s}$ will always select
    $H_i$ (because $S_i = 1$ by assumption) and the
    first $\lfloor  \Gamma n_2(\bm D) \rfloor - 1$ hypotheses corresponding to
    the p-values $\tilde{P}_{1}, \ldots, \tilde{P}_{\lfloor \Gamma
      n_2(\bm D)
      \rfloor -1}$.
    This shows that the extended selection rule $\tilde{\bm s}\circ
    \tilde{\bm d}$ is stable.
\end{proof}

\begin{example}(Counterexample of the composition of stochastically stable procedures)\label{exam:counterexample.composite.stochastic.stable}
    We follow the same notations as in the proof of \Cref{prop:composite.extended}.
    Suppose $m=3$ and we have data $\bm X = (x_1, x_2, x_3)$, where $x_1, x_2, x_3 \in (0,1)$.
    Let $\bm s$ is a random selection with probability $\frac{1}{2}$ to select task $i$. And $\bm s'$ is the BH procedure with level $q$. Then we consider $\pr(\bm s'\circ \bm s(\bm X) = (1,0, 0)\mid \tilde S_1 = 1, \bm X)$.

    When $\tilde S_1 = 1$, we must have $\bm s(\bm X) = (1,0,0), (1,1,0), (1,0,1)$ or $(1,1,1)$. Suppose $\bm X = (\frac{3}{4}q, \frac{1}{2}q, 1), \bm X' = (\frac{7}{12}q, \frac{1}{2}q, 1)$. One can check that
    \begin{align*}
        \pr(\bm s'\circ \bm s(\bm X) = (1,0,0)\mid \tilde S_1 = 1, \bm X) &= \frac{1}{2},  \\
        \pr(\bm s'\circ \bm s(\bm X') = (1,0, 0)\mid \tilde S_1 = 1, \bm X')  &=  \frac{1}{3}.
    \end{align*}

\end{example}

\subsection{\Cref{sec:multiple-risks}}
\label{app:sec:multiple-risks}




\begin{proof}[Proof of \Cref{prop:FDR.curve}]
    For each $i$, define $c^*_i := \argmin \{q(c): H_{i,c} = 0\}$. We let $c^*_i = -\infty$ if $H_{i,c} \neq 0$ for $c \in \calI$.
    Notice that,
    \begin{align*}
        \sup_{c \in \calI}\frac{\sum_{H_{i,c}  = 0} \1_{\{i \in \calS\}}}{q(c) (1 \vee |\calS|)}
        &= \sup_{c \in \calI} \sum_{i=1}^{\No} \frac{ \1_{\{H_{i,c} = 0, i \in \calS\}}}{q(c) (1 \vee |\calS|)}
        \le \sum_{i=1}^{\No} \sup_{c \in \calI}  \frac{ \1_{\{H_{i,c} = 0, i \in \calS\}}}{q(c) (1 \vee |\calS|)}\\
        &= \sum_{c_i^* \neq -\infty} \frac{ \1_{\{H_i(c_i^*) = 0, i \in \calS\}}}{q(c_i^*) (1 \vee |\calS|)}
        \le \sum_{c_i^* \neq -\infty} \frac{ \1_{\{H_i(c_i^*) = 0, P_i(X_i;c_i^*) \le q(c_i^*)|\calS|/\No\}}}{q(c_i^*) (1 \vee |\calS|)}.
    \end{align*}
    Finally, we apply the proof of \Cref{theo:post.selective.risk} to $P_i(X_i; c_i^*)$ with the selected task set $\calS$ and finish the proof.
\end{proof}

\begin{proof}[Proof of \Cref{prop:q(c).free}]
    Let $\calS_0$ denote the rejection set of applying BH to $P_i(0) = p_i(X_i, 0)$, and $\calS_{\sup}$ denote that of  \Cref{algo:FDR.curve} with $q^*(c)$. Since $q^*(0) = q$ and $q P_{i, \sup} \ge q p_i(X_i, 0)/q^*(0) = p_i(X_i, 0)$, then $\calS_{\sup} \subseteq \calS_0$.
    If $\calS_0 = \emptyset$, then $\calS_{\sup} = \calS_0$.
    If $\calS_0 \neq \emptyset$ and $i \in \calS_0$, then $p_i(X_i, 0) \le q|\calS_0|/\No$.    
    As $p(x, c)$ increases in $x$, let
    $p_i(\underline{\tau}_i, 0) = q/\No$,
    $p_i(\overline{\tau}_i, 0) = q$, then $X_i \le \overline{\tau}_i$ and 
    \begin{equation}\label{proof:prop:q(c).free.1}
    \begin{aligned}
        P_{i, \sup}
        =&\sup_{c \in \calI} \frac{p_i(X_i, c)}{q^*(c)} =   \sup_{c \in \calI} \frac{p_i(X_i, c)}{ \sup_{i\in [\No]} \sup_{\underline{\tau}_i \le x_i 
        \le \overline{\tau}_i} q\frac{p_i(x_i, c)}{p_i(x_i, 0)}} \\ 
         \le& \sup_{c \in \calI} \frac{p_i(X_i, c)}{q \cdot \sup_{\underline{\tau}_i \le x_i 
        \le \overline{\tau}_i} \frac{p_i(x_i, c)}{p_i(x_i, 0)}} \\
        \le& \sup_{c \in \calI} \frac{p_i(X_i \vee \underline{\tau}_i,c)}{q \cdot \frac{p_i(X_i \vee \underline{\tau}_i, c)}{p_i(X_i \vee \underline{\tau}_i, 0)}}
        = \frac{p_i(X_i \vee \underline{\tau}_i, 0)}{q}
        \le \frac{|\calS_0|}{\No},
    \end{aligned}
    \end{equation}
    where we use the rejection rule of BH in the last inequality. Eq.~\eqref{proof:prop:q(c).free.1} implies that
    \begin{align*}
       \left\{i: P_{i, \sup} \le \frac{|\calS_0|}{\No}  \right\}
        \supseteq \left\{i: p(X_i,0)/q \le \frac{|\calS_0|}{m} \right\}
        = \calS_0,
    \end{align*}
    where we use the property of the rejection set of the BH in the last equality.
    Together with $\calS_{\sup} \subseteq \calS_0$, we have $\calS_{\sup} = \calS_0$.

    For arbitrary $\calI'$, we first  define:
    \begin{align*}
        P'_{i,\sup} = \sup_{c\in \calI'} \frac{p_i(X_i,c)}{q(c)}, \quad P_{i,\sup} =  \sup_{c\in \calI} \frac{p_i(X_i,c)}{q^*(c)}.
    \end{align*}
    And let $\calS_0$ denote the rejection set of applying BH to $P'_{i,\sup}$ an ``adjusted FDR level'' $q = 1$, and $\calS_{\sup}$ denote that of  \Cref{algo:FDR.curve} with $q^*(c)$.

    Because $q^*(c) \le q(c)$, $\calI' \subseteq \calI$, we have $P'_{i,\sup} \le P_{i,\sup}$, thus $\calS_{\sup} \subseteq \calS_0$. To prove $\calS_{0} \subseteq \calS_{\sup}$, if $i \in \calS_0$, then $P'_{i, \sup} \le \frac{|\calS_0|}{m}$. Define  $\underline{\tau}^{c'}_i$ and $\overline{\tau}^{c'}_i$  such that $p_i(\underline{\tau}^{c'}_i, c') = q(c')/\No$,
    $p_i(\overline{\tau}^{c'}_i, c') = q(c')$, we have
    \begin{align*}
        P_{i, \sup}
        =&\sup_{c \in \calI} \frac{p_i(X_i, c)}{q^*(c)} =   \sup_{c \in \calI} \frac{p_i(X_i, c)}{\inf_{c'\in \calI'} \sup_{i\in [\No]} \sup_{\underline{\tau}^{c'}_i \le x_i 
        \le \overline{\tau}^{c'}_i}  q(c')\frac{p_i(x_i, c)}{p_i(x_i, c')}}\\
        =& \sup_{c \in \calI} \sup_{c'\in \calI'} \frac{p_i(X_i, c)}{ \sup_{i\in [\No]} \sup_{\underline{\tau}^{c'}_i \le x_i 
        \le \overline{\tau}^{c'}_i}  q(c')\frac{p_i(x_i, c)}{p_i(x_i, c')}}\\
        =& \sup_{c'\in \calI'} \sup_{c \in \calI}  \frac{p_i(X_i, c)}{ \sup_{i\in [\No]} \sup_{\underline{\tau}^{c'}_i \le x_i 
        \le \overline{\tau}^{c'}_i}  q(c')\frac{p_i(x_i, c)}{p_i(x_i, c')}}. 
    \end{align*}
    It suffices to show for any $c' \in \calI'$,
    \begin{align*}
         \sup_{c \in \calI}  \frac{p_i(X_i, c)}{ \sup_{i\in [\No]} \sup_{\underline{\tau}^{c'}_i \le x_i 
        \le \overline{\tau}^{c'}_i}  q(c')\frac{p_i(x_i, c)}{p_i(x_i, c')}} 
        \le \frac{|\calS_0|}{m}.
    \end{align*}
    Similar to Eq.~\eqref{proof:prop:q(c).free.1}, we obtain
    \begin{align*}
       &\quad~\sup_{c \in \calI} \frac{p_i(X_i, c)}{ \sup_{i\in [\No]} 
       \sup_{\underline{\tau}_i^{c'} \le x_i 
        \le \overline{\tau}_i^{c'}} q(c')\frac{p_i(x_i, c)}{p_i(x_i, c')}} 
        \le \frac{p_i(X_i \vee \underline{\tau}_i^{c'}, c')}{q(c')}\\
        &=  \frac{p_i(X_i, c')}{q(c')} \vee \frac{p_i(\underline{\tau}_i^{c'}, c')}{q(c')}
        \le P_{i, \sup}' \vee \frac{q(c')/m}{q(c')}
        \le \frac{|\calS_0|}{\No},
    \end{align*}
    this completes our proof.

\end{proof}

\subsection{\Cref{sec:accelerate.permutation}}
\label{sec:crefs}

\begin{proof}[Proof of \Cref{prop:permutation}]
We adapt the proof of \Cref{theo:post.selective.risk}.
Let $\tilde{X}_i = (P_i(M_r))_{r \in [\No]}$, then the selected set $\bm S$ of \Cref{algo:permutation} satisfies $\bm S_{-i} \perp\!\!\!\!\perp_v \tilde{X}_{i} \mid S_i = 1, \tilde{\bm X}_{-i}$.
Since $X_i \independent \bm X_{-i}$ for $H_i = 0$ and the permutations are independently generated across all tasks, then 
\begin{align*}
    T_i, (T_{ij})_{j \in [M_1]} \independent  \bm T_{-i}, (\bm T_{-i,j})_{j \in [M_1]}, 
\end{align*}
which further implies $\tilde{X}_i \independent \bm \tilde{X}_{-i}$.

We let $\tilde{q} := q/m$.
For task $i$ with $H_i = 0$,
    \begin{align*}
        \EE\left[\frac{S_i \ell_i( d_i(\bm S, \tilde{X}_i;|\calS|\tilde{q} ), \theta)}{1 \vee |\calS|}\right]
        &= \EE\left[\frac{S_i \ell_i\left( d_i(\bm S, \tilde{X}_i;|\calS|^{(-i)}\tilde{q} ), \theta\right)}{1 \vee {|\calS|}^{(-i)}}\right] 
        \tag*{$(\text{Definition of}~|\calS|^{(-i)})$}
        \\
        &\le \EE\left[\frac{ \ell_i\left( d_i(\bm S, \tilde{X}_i;|\calS|^{(-i)}\tilde{q} ),\theta\right)}{1 \vee {|\calS|}^{(-i)}}\right] 
        \tag*{$(S_i \le 1, \ell_i \ge 0)$}
        \\   &=\EE\left[\EE\left[\frac{\ell_i\left( d_i(\bm S, \tilde{X}_i;|\calS|^{(-i)}\tilde{q} ),\theta\right)}{1 \vee |\calS|^{(-i)}} \mid |\calS|^{(-i)}\right] \right] 
        \tag*{(\text{Tower property})}
        \\
        &=\EE\left[\frac{\EE\left[\ell_i\left( d_i(\bm S, \tilde{X}_i;|\calS|^{(-i)}\tilde{q} ),\theta\right)\mid |\calS|^{(-i)}\right]}{1 \vee |\calS|^{(-i)}}  \right]\\
        &\le \EE\left[\frac{|\calS|^{(-i)} \tilde{q}}{1 \vee |\calS|^{(-i)}}\right] 
        \tag*{$(|{\calS}|^{(-i)}(\tilde{\bm X}_{-i}) \independent \tilde{X}_i, d_i~\text{controls $\ell_i$-risk})$}  
        \\
        &\le \tilde{q}.
    \end{align*}
Sum over $i \in [\No]$, and we have finished the proof of post-selection risk control. Then by the same argument of \Cref{coro:extra.selective.risk}, \Cref{algo:permutation} controls the FDR (extra-selective risk).
\end{proof}

\begin{proof}[Proof of \Cref{prop:permutation.power}]
    Given a hypothesis $H_i$, let $X_{il}$, $l \in [N_i!]$ be i.i.d. Bernoulli with parameter $P_i^\infty \in [0,1]$.
    If $i \in \calS^\infty(q/(1+\delta)$, then $P_i^\infty \le q / (1+\delta)$, and further there exists $\kappa_i \in [\No]$ such that $q/(1+\delta) \cdot (\kappa_i-1)/\No < P_i^\infty \le q/(1+\delta) \cdot
 \kappa_i/\No$.
    According to the BH procedure,
    \begin{align*}
        \left\{\calS^\infty(q/(1+\delta)) \subseteq \calS(q; (M_{r})_{r \in [\No]})\right\}
        &\supseteq
        \left\{\sup_{1 \le r \le \kappa_i} P_i(M_r) - \frac{q r}{\No} \le 0, ~\text{for all}~P_i^\infty \le \frac{q}{1+\delta} \right\} \\
        &=
        \left(\bigcup_{P_i^\infty \le \frac{q}{1+\delta}}\left\{\sup_{1 \le r \le \kappa_i} P_i(M_r) - \frac{q r}{\No} > 0 \right\}\right)^c,
    \end{align*}
    where for an event $\calA$, $\calA^c$ denotes its complement.

    Next, we provide an upper bound of the probability of the event $\{\sup_{1 \le r \le \kappa_i} P_i(M_r) - {q r}/{\No} > 0 \}$.
    Let $C_{q, \varepsilon, \No, \delta} := (2/q) \cdot \left(\log\left({1}/{\varepsilon}\right) + \log(\No) \right)\cdot {(1 + 4\delta/3 + \delta^2/3)}/{\delta^2}$, then $M_r = C_{q, \varepsilon, \No, \delta} \cdot \No/r$,
    \begin{align*}
        \PP\left(\sup_{1 \le r \le \kappa_i} P_i(M_r) - \frac{q r}{\No} > 0 \right)
        &= \PP\left(\exists r \in [\kappa_i], \frac{1 + \sum_{l=1}^{M_r} X_{il}}{1 + M_r} - \frac{q r}{\No} > 0 \right) \\
        &\le \PP\left(\exists r \in [\kappa_i], {1 + \sum_{l=1}^{M_r} X_{il}} \ge \frac{q r M_r}{\No} \right) \\
        &= \PP\left(\exists r \in [\kappa_i], {1 + \sum_{l=1}^{M_r} X_{il}} \ge {q C_{q, \varepsilon, \No, \delta}}\right) \\
        &= \PP\left({1 + \sum_{l=1}^{M_{\kappa_i}} X_{il}} \ge {q C_{q, \varepsilon, \No, \delta}}\right) \\
        &= \PP\left({\sum_{l=1}^{M_{\kappa_i}} \left(X_{il} - P_{i}^\infty \right)} \ge {q C_{q, \varepsilon, \No, \delta}} - M_{\kappa_i} P_i^\infty - 1\right) \\
        &\le \PP\left({\sum_{l=1}^{M_{\kappa_i}} \left(X_{il} - P_{i}^\infty \right)} \ge \frac{\delta q C_{q, \varepsilon, \No, \delta}}{1+\delta} - 1 \right),
    \end{align*}
    where we use $M_{\kappa_i} P_i^\infty \le q C_{q, \varepsilon, \No, \delta}/(1+\delta)$ in the last inequality.
    By the Bernstein's inequality, write $t := {\delta q C_{q, \varepsilon, \No, \delta}}/{(1+\delta)} - 1$, we obtain
    \begin{align*}
        \PP\left({\sum_{l=1}^{M_{\kappa_i}} \left(X_{il} - P_{i}^\infty \right)} \ge t \right)
        \le e^{-\frac{t^2/2}{t/3 + M_{\kappa_i} P_i^\infty(1-P_i^\infty)}}
        \le e^{-\frac{t^2/2}{t/3 + (1+t)/\delta}}
        \le \frac{\varepsilon}{\No},
    \end{align*}
    where the last equality comes from the definition of $C_{q, \varepsilon, \No, \delta}$ and that $\No \ge 2$, $\varepsilon \le 0.5$, $\delta \le 1$.
    Finally, we employ the union bound and arrive at
    \begin{align*}
         \PP\left(\calS^\infty(q/(1+\delta)) \subseteq \calS(q; (M_{r})_{r \in [\No]})\right)
         \ge 1 - \sum_{i=1}^\No  \PP\left(\sup_{1 \le r \le \kappa_i} P_i(M_r) - \frac{q r}{\No} > 0 \right)
         \ge 1 - \varepsilon.
    \end{align*}

    For the computation complexity, let $T$ denote the total number of iterations, and $r_t$ denote the size of the selected set input to the $t$-th iteration, $r_1 = \No$.
    Then the total number of permutations takes the form,
    \begin{align*}
        r_1 M_{r_1} + \sum_{t=2}^T r_t (M_{r_t} - M_{r_{t-1}})
        &= r_T M_{r_T} + \sum_{t=2}^{T-1} (r_{t-1} - r_{t}) M_{r_{t-1}}\\
        &= r_T M_{r_T} + \sum_{t=2}^{T-1} \sum_{l =  r_{t} + 1}^{ r_{t-1}}  M_{r_{t-1}}\\
        &\le \sum_{l=1}^{r_T} M_{l} + \sum_{t=2}^{T} \sum_{l =  r_{t} + 1}^{ r_{t-1}}  M_{l}
        =  \sum_{l=1}^{\No} M_{l} 
        \tag*{$ (\text{$M_r$ decreases in $r$} )$}
        \\
        &= C_{q, \varepsilon, \No, \delta} \cdot \No \cdot \sum_{l=1}^{\No} \frac{1}{l}
        \le C_{q, \varepsilon, \No, \delta}  \cdot \No \cdot \left(\log(\No) + 1\right).
    \end{align*}
\end{proof}

\begin{example}\label{exam:permutation.no.lower.bound}
    Consider $P_i^\infty = q/(1+\delta)$, $i \in [\No]$.
    Let $M_i$ denote the number of permutations consumed by task $i$, $P_i(M_i)$ be the corresponding independent permutation p-values, and $\calS(q; (M_{i})_{i \in [\No]})$ be the rejection set of applying BH to $P_i(M_i)$ at level $q$.  Without loss of generality, we assume $M_i$ is increasing.
    Let $D(p; p') = p \log(p/p') + (1-p)\log((1-p)/(1-p'))$, and $C(\delta, q) = 1/D(q; q/(1+\delta))$. Here $C(\delta, q) \le 6/(q \delta^2)$ for $\delta \le 1/2$.
    If the total number of permutations $\sum_i M_i \le C(\delta, q) \No \log(\No)$, then $M_i \le 2 C(\delta, q) \log(\No)$ for all $i \le \No/2$.
    We use $C' := 2 \cdot C(\delta, q)$ for notation simplicity.

    For $q < 1$, by the anti-concentration bound of Binomial distribution,
    \begin{align*}
        &\quad~\PP\left(P_i(M_i) = \frac{1 + \sum_{l=1}^{M_i} X_{il}}{1 + M_i} > q \right)
        \ge \PP\left(\frac{\sum_{l=1}^{M_i} X_{il}}{M_i} \ge q \right) \\
        &\ge \frac{1}{\sqrt{2 M_i}} e^{-M_i D\left(q; \frac{q}{1+\delta}\right)}
        \ge \frac{1}{\sqrt{2 C'\log(\No)}} e^{-C'\log(\No) D\left(q; \frac{q}{1+\delta}\right)}\\
        &= \frac{1}{\sqrt{2 C'\log(\No)}} \cdot \frac{1}{\No^2}
        \ge \frac{2\log(1/\varepsilon)}{\No},
    \end{align*}
    for $\No$ sufficiently large depending on $\delta$, $\varepsilon$, and $q$.
    Notice that applying BH to $P_i^{\infty}$ at level $q/(1+\delta)$ rejects all the hypotheses, then
\begin{align*}
    &\quad~\PP\left(\calS^\infty(q/(1+\delta)) \not\subseteq \calS\left(q; (M_{i})_{i \in [\No]}\right)\right)\\
    &\ge \PP\left(\exists i \le \frac{\No}{2}, P_i(M_i) > q \right)
    = 1 - \PP\left(\forall i \le \frac{\No}{2}, P_i(M_i) \le q \right) \\
    &= 1 - \PP\left(P_i(M_i) \le q \right)^{\frac{\No}{2}}
    \ge 1 - \left(1 - \frac{2\log(1/\varepsilon)}{\No}\right)^{\frac{\No}{2}}
    \ge 1 - \varepsilon.
\end{align*}
This implies at least a total of $C(\delta, q) \No \log(\No)$ permutations are required.
\end{example}

\subsection{\Cref{appe:sec:stable}}

\begin{proof}[Proof of \Cref{sufficient condition for extended stable rule}]
  Denote $\bm s = \tilde{\bm s} \circ \tilde{\bm d}$.
  The assumptions mean that, with an abuse of notation,
  $\tilde{d}_i(\cdot, \bm X) =
  \tilde{d}_i(\cdot, X_i)$ and $\tilde{s}_i(\bm D, \bm X) =
  \tilde{s}_i( D_i, X_i)$. Thus, $s_i(\cdot, \bm X) = s_i(\cdot, X_i)$. It
  follows from the definition that $\bm s$ is stable. 
\end{proof}

\begin{proof}[Proof of \Cref{prop:composite.extended} (1)]
     We prove for  stable $\bm s'$, and the proof is similar when $\bm s'$ is weakly stable.
     Let $\bm X'$ satisfy that $\bm X_{-i}' = \bm X_{-i}$.
     Since $\bm s$ is stable, given $s_i(\bm X) = s_i(\bm X') = 1$, we have $s_j(\bm X) = s_j (\bm X'), j\in [\No]$. Since $\bm s'$ is stable, for any $\bm r\in \{0,1\}^m$, given $s'_i(\bm r,\bm X) = s'_i(\bm r, \bm X') = 1$, we have $s'_j(\bm r, \bm X) = s'_j (\bm r, \bm X'), j\in [\No]$.
    Let $\tilde{S}_i$ denote whether task $i$ is selected by the composition $\bm s' \circ \bm s$, then $\tilde S_i = 1$ implies $s_i(\bm X) = 1$ and $s'_i(\bm s(\bm X), \bm X) = 1$.
    Therefore, given $\tilde S_i= 1$ and $\bm X_{-i}' = \bm X_{-i}$,
    \begin{align*}
        \bm s'\circ \bm s (\bm X) = \bm s'(\bm s(\bm X), \bm X)
        = \bm s'(\bm s(\bm X'), \bm X)
        = \bm s'(\bm s(\bm X'), \bm X')
        = \bm s'\circ \bm s (\bm X'),
    \end{align*}
    meaning that the composition $\bm s'\circ \bm s$ is stable.
\end{proof}

\begin{proof}[Proof of \Cref{prop:composite.extended} (2)]
We first prove that    if $\bm s$, $\bm s'$ are stochastically stable, that is $\bm s$ and $\bm s'$ are random functions of $\bm X$ (which contain deterministic functions as a special case), then $\bm s \cap \bm s'$ is stochastically stable.

    We denote $\bm S = \bm s(\bm X)$ and $\bm S' = \bm s'(\bm X)$. From the definition of $\bm s$ and $\bm s'$,  we have
    \begin{align*}
        \pr(\bm S = \bm w_1\mid S_i = 1, X_i, \bm X_{-i}) &= \pr(\bm S = \bm w_1\mid S_i = 1, \bm X_{-i}), \\
        \pr(\bm S' = \bm w_2\mid S'_i = 1, X_i, \bm X_{-i}) &= \pr(\bm S' = \bm w_2\mid S'_i = 1, \bm X_{-i}).
    \end{align*}
    then for the intersection of $\bm s$ and $\bm s'$, notice that $\bm s(\bm X)
    \independent \bm s'(\bm X) \mid \bm X$, we have
    \begin{align*}
    &\pr(\bm s \cap \bm s'(\bm X) = \bm w \mid  S_i=1, S'_i = 1, \bm X)  \\
    =&\frac{\pr(\bm s \cap \bm s'(\bm X) = \bm w, S_i=1, S'_i = 1, \bm X)}{\pr( S_i=1, S'_i = 1, \bm X)}  \\
    =& \frac{\sum_{\bm w_1,\bm w_2 ~s.t. \bm w_1\bm w_2 = \bm w }\pr(\bm S = \bm w_1, \bm S' = \bm w_2, S_i=1, S'_i = 1, \bm X)}{\pr( S_i=1, S'_i = 1, \bm X)} \\
    =& \frac{\sum_{\bm w_1,\bm w_2 ~s.t. \bm w_1\bm w_2 = \bm w}\pr(\bm S = \bm w_1,  S_i=1\mid  \bm X) \pr(\bm S' = \bm w_2, S'_i = 1\mid \bm X)}{\pr( S_i=1\mid \bm X)\pr(S'_i =1\mid \bm X )} \\
    =& \sum_{\bm w_1,\bm w_2 ~s.t. \bm w_1\bm w_2 = \bm w} \pr(\bm S = \bm w_1\mid S_i=1, \bm X) \pr(\bm S' = \bm w_2\mid S'_i = 1, \bm X)  \\
    =& \sum_{\bm w_1,\bm w_2 ~s.t. \bm w_1\bm w_2 = \bm w} \pr(\bm S = \bm w_1\mid  S_i=1, \bm X_{-i}) \pr(\bm S' = \bm w_2\mid S'_i = 1, \bm X_{-i}),
    \end{align*}
which does not depend on $X_i$. So $\bm s \cap \bm s'$ is stochastically stable. In particular, when $\bm s$ and $\bm s'$ are stable, $\bm s \cap \bm s'$ are stochastically stable. Because $\bm s  \cap \bm s'$ is deterministic, $\bm s \cap \bm s'$ is stable.
\end{proof}

\begin{proof}[Proof of \Cref{prop:arbitrary dependence}]
    The proof is adapted from \textcite{benjamini2001control}. We
    define $\tilde q = q/f(m)$ which
    only concerns FDR (extra-selection risk). Consider any $i$ such
    that $H_i = 0$ ($H_i$ is true). We have
        \begin{align*}
            &\quad~\EE\left[\frac{S^1_i \ell_i( d_i(\bm S^1, X_i;|\calS^1|\tilde{q} ), \theta)}{1 \vee {|\calS^1|}}\right] 
            \tag*{(risk attributed to task $i$)}
            \\
            &\le \EE\left[\frac{ \ell_i( d_i(\bm S^1, X_i;|\calS^1|\tilde{q} ), \theta)}{1 \vee {|\calS^1|}}\right] 
            \tag*{$(S^1_i \le 1, \ell_i \ge 0)$}
            \\
            &= \EE \left[ \sum_{k=1}^{m} \frac{\1_{ \{ {|\calS^1|} = k \} }}{k}\ell_i( d_i(\bm S^1, X_i;k\tilde{q} ), \theta)  \right] 
            \tag*{(enumerate over values of $|\calS^1|$)} 
            \\
            &= \EE \left[ \sum_{k=1}^{m} \frac{\1_{ \{ {|\calS^1|} = k \} }}{k} \sum_{j=1}^k \left(\ell_i( d_i(\bm S^1, X_i;j\tilde{q} ), \theta)
            -\ell_i( d_i(\bm S^1, X_i;(j-1)\tilde{q} ), \theta)\right)\right] 
            \tag*{(telescoping sum and~$\ell_i( d_i(\bm S^1, X_i;0), \theta) = 0 ~\text{by}~ r_i( d_i(\bm S^1, X_i;0), \theta) = 0$)} 
            \\
            &= \EE \left[ \sum_{j=1}^{m} \left(\ell_i( d_i(\bm S^1, X_i;j\tilde{q} ), \theta)
            -\ell_i( d_i(\bm S^1, X_i;(j-1)\tilde{q} ), \theta)\right) \sum_{k=j}^{m} \frac{\1_{ \{ {|\calS^1|} = k \} }}{k} \right] 
            \tag*{(switching the order of summation indexed by~ $k, j$)}
            \\
            &\le \EE \left[ \sum_{j=1}^{m} \left(\ell_i( d_i(\bm S^1, X_i;j\tilde{q} ), \theta)
            -\ell_i( d_i(\bm S^1, X_i;(j-1)\tilde{q} ), \theta)\right) \cdot \frac{1}{j} \right] 
            \tag*{($\ell_i(d_i(\cdot, q), \theta)$ is non-decreasing in $q$ and~ $\sum_{j=k}^{m} \1_{ \{ {|\calS^1|} = k \}} \le 1$)}
            \\
            &= \EE \left[\ell_i( d_i(\bm S^1, X_i;m\tilde{q} ), \theta) \cdot \frac{1}{m} + \sum_{j=1}^{m-1}
            \ell_i( d_i(\bm S^1, X_i;j\tilde{q} ), \theta) \cdot \left(\frac{1}{j} - \frac{1}{j+1}\right)
            \right] 
            \tag*{(\text{summation by parts and}~ $\ell_i( d_i(\bm S^1, X_i;0), \theta) = 0$)}
            \\
            &\le \tilde q + \sum_{j=1}^{m-1}j\tilde q  \cdot \frac{1}{j(j+1)} 
            \tag*{($d_i$~\text{controls $\ell_i$-risk})}
            \\
            &= \frac{q}{m}.
        \end{align*}
        Finally, we sum over $H_i = 0$ for $i \in [\No]$ and finish the proof.
\end{proof}

\begin{proof}[Proof of \Cref{prop:iterate.equivalence}]
  Let $T_{\text{par}},T_{\text{seq}} \leq \No $ be the number of
  iterations for parallel intersection and sequential composition to
  converge, respectively. Then the claim is that
  \[
    \bm S_{\text{par}}^\ast:=(\tilde{\bm s}_{\text{par}} \circ \tilde{\bm
      d}_{\text{par}})^{T_{\text{par}}}(\bm S^1) = (\bm s_{\text{seq}})^{T_{\text{seq}}}(\bm S^1) =: \bm S_{\text{seq}}^\ast,
  \]
  and the decisions are the same, too.
  
  Under the contracting and  increasing property of $\tilde{\bm s}_{c}\circ \tilde{\bm d}_{c}, \forall c\in [k]$, we have $\bm s_{\text{seq}}(\bm S) \preceq \tilde{\bm s}_{\text{par}}\circ \tilde{\bm d}_{\text{par}}(\bm S')$ if $\bm S\preceq \bm S'$. 
  As the two algorithms start with the same selected set $\bm S^1$,  we have $\bm{S}^t_{\text{seq}} \preceq \bm{S}^t_{\text{par}}, \forall t \geq 1$. Choosing $t=m$, we have $\bm{S}^m_{\text{seq}} \preceq \bm{S}^m_{\text{par}}$. Because $\bm{S}^\ast_{\text{seq}} = \bm{S}^{T_{\text{seq}}}_{\text{seq}} = \bm{S}^m_{\text{seq}}$ and $\bm{S}^\ast_{\text{par}} = \bm{S}^{T_{\text{par}}}_{\text{par}} = \bm{S}^m_{\text{par}}$, we get $\bm{S}^\ast_{\text{seq}} \preceq \bm{S}^\ast_{\text{par}}$. 
  While for the other side, because $\bm{S}^\ast_{\text{par}}$ is a fixed point of $\tilde{\bm s}_{\text{par}}\circ \tilde{\bm d}_{\text{par}}$, i.e. $\tilde{\bm s}_{\text{par}}\circ \tilde{\bm d}_{\text{par}}(\bm{S}^\ast_{\text{par}}) = \bm{S}^\ast_{\text{par}}$, which means that $\tilde{\bm s}_c\circ \tilde{\bm d}_c(\bm{S}^\ast_{\text{par}}) = \bm{S}^\ast_{\text{par}}, \forall c\in [k]$. Then by definition we have $\bm s_{\text{seq}}(\bm{S}^\ast_{\text{par}}) = \bm{S}^\ast_{\text{par}}$. So $\bm{S}^\ast_{\text{par}}$ is also a fixed point of $\bm s_{\text{seq}}$. 
  Then because $\bm{S}^\ast_{\text{par}} \preceq \bm S^1 =\bm S^1_{\text{seq}} $, applying the extended selection rule $\bm s_{\text{seq}}$ to both sides, we have $\bm{S}^\ast_{\text{par}} = \bm s_{\text{seq}}(\bm{S}^\ast_{\text{par}}) \preceq \bm s_{\text{seq}}(\bm S^1) = \bm S^2_{\text{seq}}$ by the increasing property of $\bm s_{seq}$. Similarly by induction, we get $\bm{S}^\ast_{\text{par}} = \bm s_{\text{seq}}(\bm{S}^\ast_{\text{par}}) \preceq \bm s_{\text{seq}}(\bm{S}^t_{\text{seq}}) = \bm{S}^{t+1}_{\text{seq}}, \forall t \geq 1$.  Therefore when the sequential composition algorithm converges ($t= T_{\text{seq}}$), $ \bm{S}^\ast_{\text{par}} \preceq \bm S^{T_{\text{seq}}}_{\text{seq}} =  \bm S^\ast_{\text{seq}} $.
\end{proof}



\end{document}